%% file: Gfast2015_l1norm.tex
\documentclass[10pt,journal]{IEEEtran}
\input{usepackages}
\input{macros}

\newcommand{\tr}[1]{\rm{Tr}\left(#1\right)}
\newcommand{\bs}{\boldsymbol}
\newcommand{\dsum}{\displaystyle\sum}
\newcommand{\dmax}{\displaystyle\max}
\newcommand{\dmin}{\displaystyle\min}

\IEEEoverridecommandlockouts

\newtheorem{observation}{Observation}
\begin{document}

\title{Deep-LMS for gigabit transmission over unshielded twisted pair cables}
\author{Avi~Zanko$^{*,\natural}$, Itsik~Bergel$^*$(Senior Member, IEEE), Amir~Leshem$^*$ (Senior Member, IEEE)
\thanks{$^*$ Faculty of Engineering, Bar-Ilan University, Ramat-Gan, 52900, Israel.}
\thanks{$^{\natural}$ Corresponding author, email: avizanz@gmail.com.}
\thanks{These results will be partially presented at EUSIPCO 2016 \cite{Zanko_Gigabit_2016}.}}
\maketitle
\begin{abstract}
\label{sec:abstract}
In this paper we propose a rapidly converging LMS algorithm for crosstalk cancellation. The architecture is similar to deep neural networks, where multiple layers are adapted sequentially. The application motivating this approach is gigabit rate transmission over unshielded twisted pairs using a vectored system. The crosstalk cancellation algorithm uses an adaptive non-diagonal preprocessing matrix prior to a conventional LMS crosstalk canceler. The update of the preprocessing matrix is inspired by deep neural networks. However, since
most the operations in the Deep-LMS algorithm are linear, we
are capable of providing an exact convergence speed analysis.
The role of the preprocessing matrix is to speed up the convergence of
the conventional LMS crosstalk canceler and hence the
convergence of the overall system. The Deep-LMS is important for crosstalk cancellation in the novel G.fast standard, where traditional LMS converges very
slowly due to the ill-conditioned covariance matrix of the received signal at the extended bandwidth. Simulation results support
our analysis and show significant reduction in convergence time
compared to existing LMS variants.
\end{abstract}
\section{Introduction}
\label{sec:introduction}
The Digital Subscriber Line (DSL) family of technologies
provides high speed digital data transmission over unshielded twisted
pairs (UTP) of copper wires that were originally used for
telephone services \cite{Bingham_Adsl_2000,Starr_Understanding_1999}. The performance of a DSL system may be
significantly degraded by the effect of crosstalk due to electromagnetic
coupling between adjacent pairs. We distinguish
between two types of crosstalk: near end crosstalk (NEXT) is the interference that is measured at the same side as the interfering transmitter. Far end crosstalk (FEXT) is the interference measured at the other end of the cable with respect to the interfering transmitter. In DSL systems, NEXT is avoided by using frequency division duplexing (FDD) \cite{Bingham_Adsl_2000}, or time division duplexing (TDD) \cite{teken}. FEXT suppression involves spectrum shaping, advanced precoding or crosstalk cancellation (also known as vectoring or DSM level $3$). A detailed surveys of these techniques can be found in e.g., \cite{Cioffi_Vectored_2006,Huberman_Dynamic_2012,Bergel_Signal_2013,Leung_Vectored_2013}.

Over the last few decades, DSL technologies have advanced in many directions to support the increasing demand for high data rates without replacing the existing copper infrastructure. Examples of such directions include increasing the bandwidth, using advanced algorithms for FEXT cancellation, and deploying fibers closer to the customer-premise equipment (CPE) so shorter copper lines can be used. For vectored VDSL modems, the downstream channel is weakly row-wise diagonally dominated, whereas for the upstream it is column-wise diagonal dominant. Hence, approximate matrix inversion is used instead of the ZF linear precoding \cite{Leshem_A_Low_2007}. Alternatively adaptive LMS based techniques for computing the precoder have been suggested \cite{Bergel_Convergence_2010,Binyamini_Arbitrary_2012,Cendrillon_Near_2007,Louveaux_Adaptive_2006,Louveaux_Adaptive_2010}.
Similarly, for the upstream, adaptive linear cancelers have proved to be very efficient \cite{Cendrillon_Partial_2004,Homer_Adaptive_2006}.

Gigabit over DSL was presented as early as 2003 \cite{Cioffi_Gigabit_2003}. Recently, the G.fast standard was introduced \cite{teken} to provide aggregate data rates of up to $1$Gb/sec over short lines up to $250$m in length. For this purpose, the bandwidth was increased from $30$MHz to $106$MHz in G.fast \cite{teken} or $212$MHz in the more advanced version. However, this increased bandwidth has led to new challenges: At high frequencies the channel matrices become non-diagonal dominant \cite{Leshem_A_Low_2007}.
Traditional DSL algorithms for crosstalk cancellation/precoding have thus become inefficient or converge very slowly, because of this significant increase in bandwidth.

This paper focuses on designing a per tone adaptive crosstalk canceler for the upstream transmissions in G.fast systems. Given the tremendous amount of data to be processed each second\footnote{In G.fast there are 2048/4096 sub-carriers and a discrete multitone (DMT) symbol rate of $48$kSymbol/sec \cite{teken}.}, low complexity algorithms are necessary. In the literature, several adaptive algorithms have been proposed that minimize the mean square error (MSE) at the output of the crosstalk canceler. Clearly, the most popular algorithm is the least mean square (LMS) algorithm \cite{Widrow_Adaptive_1960}. LMS is a stochastic optimization method that under certain conditions converges to the Minimum MSE (MMSE) solution. The popularity of the LMS algorithm derives primarily from its simplicity. However, it is well known that bad conditioning of the input correlation matrix may lead to slow convergence of the algorithm.

Because it is a stochastic optimization algorithm, the rate of convergence of the LMS algorithm is controlled by the step-size parameter $\mu$. While small step size leads to better precision of the steady state solution, a large step size is preferable for short transient state (fast convergence). However, if the step size is too large the system becomes unstable, and the algorithm may not converge at all. Thus, even for an optimal choice of the step size, the convergence rate strongly depends on the statistics of the input signal: the conventional LMS algorithm uses a fixed step-size parameter $\mu$ to control the correction applied to the solution at each iteration. An exponential convergence in the mean and in the mean square were shown if the step size is bound by a term that is inversely proportional to the largest eigenvalue and to the trace (respectively) of the input signal's correlation matrix and the rate of convergence was shown to be determined by the spread of the eigenvalues of the correlation matrix \cite{Widrow_Adaptive_1971,Horowitz_Performance_1981,Feuer_Convergence_1985}.

Several LMS derivatives exhibit a faster convergence.
In the normalized LMS (NLMS)\cite{Nagumo_A_learning_1967}, a step-size that is normalized with the power of the input is used (for low power inputs, a certain regulation parameter is typically used). In the diagonal step-size matrix approach, a diagonal matrix is used instead of the scalar step size $\mu$ \cite{Harris_A_variable_1986}. Another family of algorithms that converge to the optimal value in the mean square uses a time variant step-size parameter such that $\mu=\mu(n)$ is made inversely proportional to the iteration number $n$ \cite{Kushner_Stochastic_2003}. However, a decrease in the step-size impairs the tracking performance of the algorithm when the optimal solution is time variant.

So far, we have assumed a single step size $\mu$ for the update of the solution. However, it is possible to use a diagonal step-size matrix ${\bf M}$ instead of $\mu$ \cite{Harris_A_variable_1986}. The step-size matrix ${\bf M}$ can be thought of as a pre-processing of the input vector of the LMS algorithm in which an appropriate diagonal matrix ${\bf M}$ results in a "new" input vector to the LMS algorithm with an improved eigenvalue spread of its correlation matrix. Such an approach was taken in \cite{Narayan_Transform_1983}, where the LMS algorithm was applied to an orthogonal transformation of the input signal (e.g., discrete Fourier transform (DFT) and discrete cosine transform (DCT)) such that different frequency components of the filter were assigned different step sizes\footnote{Note that the orthogonal transformation by itself does not change the eigenvalue spread of the correlation matrix.}. The Leaky LMS algorithm is similar to the conventional LMS but minimizes a slightly modifying cost function \cite{Mayyas_Leaky_1997}. This modification can be interpreted as an addition of white noise with a certain variance to the input of the LMS (also known as pre-whitening). If this variance is appropriately chosen, there is improvement in the convergence rate. This improvement can be explained by the fact that the algorithm can be recast as a conventional LMS with a decreased eigenvalue spread.

Another framework of solutions to accelerate the LMS utilize a variety of averaging methods of the LMS coefficients (see for example \cite{Kushner_Stochastic_1993}). In the accelerated LMS, the LMS component are updated off-line as in the conventional LMS, but the actual filter that is applied to the data at each iteration $n$ is a certain averaged version of previous filters. It should be mentioned that the averaging method showed no asymptotic improvement; i.e., when a decreasing step-size is used, the optimal rate of convergence of both systems was the same.

A study of two adaptive filters in tandem \cite{Ho_A_study_2000} showed that the convergence rate of the system decreased compared to a single adaptive filter and increased the variance in the steady state. An extension of the result for multiple filters in tandem was given in \cite{Ho_Performance_2001}.
The tandem scheme is fundamentally different from the deep-LMS algorithm. As will be shown in what follows, the deep-LMS algorithm serially concatenates the blocks. Therefore, the equivalent scheme is the product of the two blocks and not the sum as in the tandem scheme \cite{Ho_Performance_2001}. Moreover, in the tandem scheme, both LMS components are updated in each iteration. In our scheme, the preprocessing matrix is updated infrequently.

As mentioned above, at high frequencies, G.fast channel matrices typically have no diagonal dominant structure. Hence, at these frequencies the received signal's covariance matrix is badly conditioned even in cases where the direct channel gains of all users are of the same order. Therefore, using the conventional LMS algorithms for FEXT cancellation in G.fast systems is impractical due to the long convergence times.

{\bf Main contributions:} In this paper, we propose a new
algorithm dubbed Deep-LMS, which enjoys almost the same complexity as the conventional LMS algorithm but exhibits a much faster convergence. The key idea behind the proposed algorithm can be summarized as follows: the asymptotic convergence of the LMS algorithm is guaranteed if the step size is properly chosen. Thus, some progress must be achieved by the LMS algorithm; i.e., by waiting a sufficiently long time. This progress is harnessed by the Deep-LMS by updating the adaptive preprocessing matrix such that the effective channel matrix becomes more diagonal dominant. Hence, the convergence of the
overall system is accelerated.

The update of the preprocessing matrix is inspired by deep neural networks. However, since all the operations in the Deep-LMS algorithm are linear, we are able to derive an exact analysis. In particular, we prove that if the initial SINR of all users is sufficiently high, the rate of convergence of the LMS crosstalk canceler is accelerated at each update of the preprocessing matrix (assuming that there are enough iterations between the updates).

The structure of this paper is as follows. Section \ref{sec:system_model} describes the system model and the proposed algorithm. Section \ref{sec:conventional_LMS} the conventional LMS crosstalk canceler is described. In Section \ref{sec:new_algorithm} the Deep-LMS algorithm is presented. In section \ref{sec:analysis} we analyze the transition of the Deep-LMS algorithm and discuss the effects on the overall system. The proof of the main theorem is given in Section \ref{sec:proof}. Section \ref{sec:numerical_results} provides simulation results for the Deep-LMS algorithm. Section \ref{sec:conclusion} concludes the paper.

\section{System model}
\label{sec:system_model}
In this section we describe the model of the received signal. For simplicity of presentation, the model is described for a single frequency bin. The operation at other frequency bins is similar and independent of the operation at any other bin.

\begin{figure}[t]
  \includegraphics[width=0.5\textwidth]{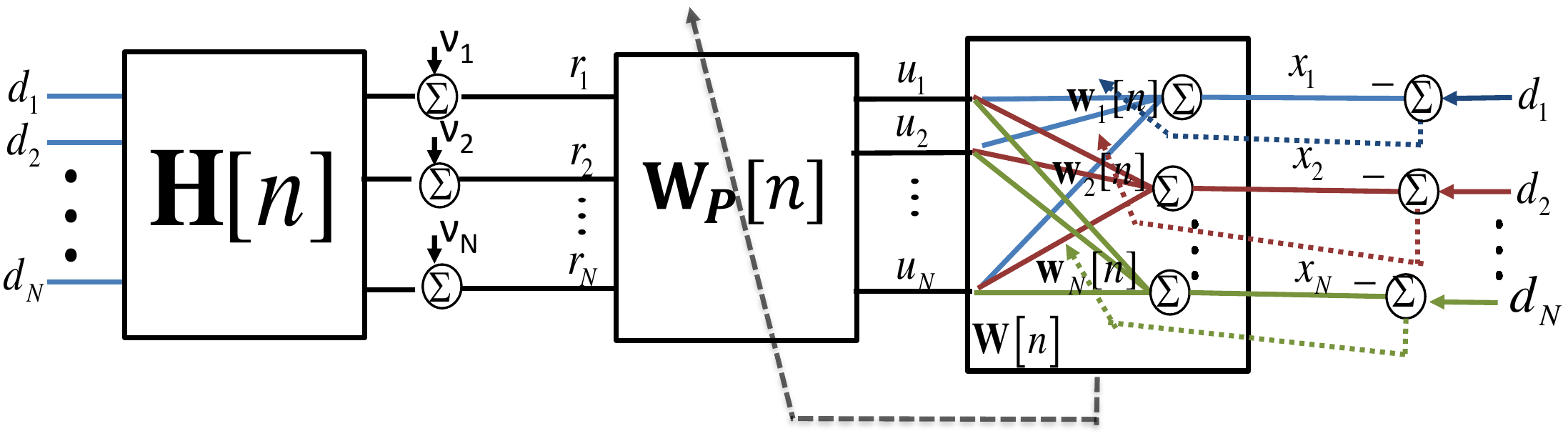}
  \caption{The Deep-LMS crosstalk canceler.}
  \label{fig:two_steps_LMS}
\end{figure}

Let ${\bf d}[n]$ denote the transmitted data symbols of all $N$ users at a certain frequency bin, where ${\bf d}[n]$ is drawn with Gaussian i.i.d components and is assumed to be known at the receiver. More formally, ${\bf d}[n]$ is a zero mean Gaussian vector with a correlation matrix of $E\{{\bf d}[n]{\bf d}^H[n]\}={\bf I}$ and $E\{{\bf d}[n]{\bf d}^H[m]\}={\bs 0}$ for $m\neq n$, where ${\bf I}$ is a $N \times N$ identity matrix and ${\bf A}^H$ is the conjugate transpose of the matrix ${\bf A}$. The received signal at the output of the channel is
\begin{equation}
\label{eq:received_signal_non_normalized}
\breve{\bf r}[n]=\sqrt{p}{\bf H}^H[n]{\bf d}[n]+\breve{\bs\nu}[n],
\end{equation}
where $p$ is the transmitted power (at that frequency bin), ${\bf H}[n]=\left[{\bf h}_{1}[n],{\bf h}_{2}[n],\cdots,{\bf h}_{N}[n]\right]$ is the channel gain matrix, $\breve{\bs\nu}[n]$ is an additive white Gaussian noise (AWGN), $\breve{\bs\nu}[n]\in\mathcal{CN}\left({\bs 0},\breve\sigma_{\nu}^2{\bf I}\right)$. For convenience, we rewrite equation (\ref{eq:received_signal_non_normalized}) by normalizing it with respect to the transmitting power:
\begin{equation}
\label{eq:r}
{\bf r}[n]={\bf H}^H[n]{\bf d}[n]+{\bs\nu}[n],
\end{equation}
such that ${\bf r}[n]=\breve{\bf r}[n]/\sqrt{p}$, ${\bs\nu}[n]\in\mathcal{CN}\left({\bs 0},\sigma_{\nu}^2{\bf I}\right)$ and $\sigma_{\nu}^2=\frac{\breve\sigma_{\nu}^2}{p}$.

\section{The conventional LMS crosstalk canceler}
\label{sec:conventional_LMS}
In this section we describe the model for a multi-user DSL system with a conventional LMS crosstalk canceler. An LMS crosstalk canceler for the upstream transmission of $N$ users can be
easily implemented by $N$ parallel LMS blocks. To avoid confusion (for reasons that will become clearer in what follows) we denote the input signal (vector) to the LMS by ${\bf u}[n]$. In the conventual approach, the LMS crosstalk canceler is directly applied on the received signal, i.e., ${\bf u}[n]={\bf r}[n]$.
The LMS crosstalk canceler (in each tone) can be described as a matrix
\begin{equation}
\label{eq:LMS_matrix}
{\bf W}[n]\triangleq\left[{\bf w}_{1}[n],{\bf w}_{2}[n],\cdots,{\bf w}_{_{N}}[n]\right].
\end{equation}
Hence, the LMS recursion can be written in a matrix form as
\begin{equation}\label{eq:LMS_matrix_recursion}
{\bf W}[n+1]={\bf W}[n]-2\mu{\bf u}[n]{{\bf e}^H[n]}
\end{equation}
and the output of the LMS crosstalk canceler can be written as
${\bf x}[n]={\bf W}^H{\bf u}[n]$.

\section{The Deep-LMS algorithm}
\label{sec:new_algorithm}
In this section, we present the novel Deep-LMS crosstalk canceler. Unlike the conventional approach, in this algorithm the LMS ${\bf W}[n]$ is not applied directly on the received signal. Instead, the received signal is preprocessed by a matrix ${\bf W}_{\rm{P}}$.
Hence, the input of the adaptive block can be written as
\begin{equation}
{\bf u}[n]={\bf W}_{\rm{P}}^H[n]{\bf r}[n].
\end{equation}
The model is illustrated in Fig~.\ref{fig:two_steps_LMS}.

The algorithm is initialized with an identity preprocessing matrix and the preprocessing matrix is updated only at certain time instances. Denote the set of all update time instances by $\mathcal{U}$. The algorithm is summarized with the following set of equations:
\begin{align}
  &{\bf W}[0] = {\bf I} \\
  &{\bf W}_{\rm{P}}[0] = \tilde{\bf D}[0] \\
  &\label{eq:Deep_LMS_u}{\bf u}[n] = {\bf W}_{\rm{P}}^H[n]{\bf r}[n] \\
  &\label{eq:Deep_LMS_x}{\bf x}[n] = {\bf W}^H[n]{\bf u}[n] \\
  &{\bf e}[n] = {\bf d}[n]-{\bf x}[n] \\
\label{eq:LMS_update_U}
  &\breve{\bf W}[n+1]={\bf W}[n]+2\mu{\bf u}[n]{\bf e}^H[n]\\
  &{\bf W}[n+1] = \begin{cases}
                \breve{\bf W}[n+1] & n\notin\mathcal{U}\\
                {\bf I} & n\in\mathcal{U}
            \end{cases}\\
  &{\bf W}_{\rm{P}}[n+1] = \begin{cases}
                {\bf W}_{\rm{P}}[n] & n\notin\mathcal{U}\\
                \label{eq:the_prefilter}{\bf W}_{\rm{P}}[n]\breve{\bf W}[n+1]\tilde{\bf D}[n+1] & n\in\mathcal{U}
            \end{cases}
\end{align}
where $\tilde{\bf D}[n]$ is the diagonal normalization matrix that satisfies $({\bf H}{\bf W}_{\rm{P}}[n])_{i,i}=1$ for all $i$. As can be seen, for $n\notin\mathcal{U}$ the preprocessing matrix remains unchanged and the LMS block ${\bf W}[n]$ is updated exactly as in the conventional LMS crosstalk canceler. For $n\in\mathcal{U}$ the preprocessing matrix ${\bf W}_{\rm{P}}[n]$ is updated to include the combined effect of the preprocessing and the current LMS, and the LMS ${\bf W}[n]$ is initiated back to ${\bf I}$.

The role of the preprocessing matrix is to speed up the convergence of the LMS crosstalk canceler ${\bf W}[n]$ between the updates, i.e., $n\notin\mathcal{U}$ and hence speed up the convergence of the overall system.

\section{Analysis of the Deep-LMS algorithm}
\label{sec:analysis}
In this section, we analyze the convergence of the Deep-LMS algorithm through a characterization of the minimal signal to interference plus noise ratio (SINR) at the input and the output of the LMS crosstalk canceler between two consecutive updates of the preprocessing matrix. The main result of this paper (Theorem \ref{theorem:main_abstract}) states that if the update time is properly chosen, the rate of improvement in the SINR at the output of the Deep-LMS will be higher after the update of the preprocessing matrix.

Formally, let $\tilde{\bf H}[n]={\bf W}_p[n]{\bf H}[n]$
denote the effective channel matrix for the LMS crosstalk canceler; i.e., after the preprocessing. The SINR at the $i$-th input of the LMS crosstalk canceler; i.e., the SINR that is measured in $u_i[n]$, see Fig. \ref{fig:two_steps_LMS} is given by:
\begin{equation}
\label{eq:SINR}
\Phi_{i}[n]=\frac{|\tilde{h}_{i,i}[n]|^2}{\dsum_{j\neq i}{|\tilde{h}_{j,i}[n]|^2}+\tilde{\sigma}_i[n]},
\end{equation}
where $\tilde{\sigma}_i[n]$ is the variance of the $i$-th entry of the colored noise $\tilde{\bs \nu}[n]={\bf W}_{\rm{P}}^H[n]{\bs\nu}[n]$. The minimal input SINR to the LMS block is defined as $\Phi[n]=\dmin_{i}\Phi_i[n]$.

In the following analysis, we focus on the set of times between two updates of the preprocessing matrix; i.e., the set of times $\mathcal{F}_\ell\triangleq \{n_\ell,n_{\ell}+1,\ldots,n_{\ell+1}-1\}$, where $n_\ell\in\mathcal{U}$ and $n_{\ell+1} =\dmin_{n\in\mathcal{U}} \{n>n_\ell\}$, i.e., $n_{\ell+1}$ is the first update of the preprocessing matrix ${\bf W}_{\rm{P}}[n]$ after time $n_{\ell}$. Note that ${\bf W}_p[n]$ and $\tilde{\bf  H}[n]$ do not change during the analyzed interval, and hence $\Phi[k]=\Phi[m]$ for any $k,m\in\mathcal{F}_\ell$. Therefore, it is convenient to denote the minimal input SINR at these times by $\Phi_\ell$. The same argument applies to the covariance matrix
\begin{equation}
\label{eq:R_ell}
    {\bf R}_{\ell}=E\{{\bf u}[n]{\bf u}^H[n]\}.
\end{equation}
In addition, for the sake of readability, the time index $n$ is omitted in cases where it is clear from the context that $n\in\mathcal{F}_\ell$. Also note that at time $n_\ell\in\mathcal{U}$, the LMS crosstalk canceler is initiated and therefore ${\bf W}[n_\ell]={\bf I}$. Hence, at the update time instance, the input SINR to the LMS block is also the output SINR of the the Deep-LMS; i.e., the SINR that is measured in $x_i[n]$.
\begin{theorem}
\label{theorem:main_abstract}
Assuming that $\mu=\frac{1}{3\tr{{\bf R}_\ell}}$, the minimal SINR at time $n_{\ell+1}$ is lower bounded by
\begin{equation}
\label{eq:main}
\Phi_{\ell+1}\geq\left(\left(c\cdot a^{n_{\ell+1}-n_\ell}\Phi_{\ell}\right)^{-1}+\eta_{\infty}(\ell)\right)^{-1}-1,
\end{equation}
where\\
$c=\bigg(1+\delta\left(\Phi_{\ell}\right)\biggr)^{-1}$, $\delta\left(\Phi\right)=\frac{1+2\alpha(\Phi)}{\Phi-\alpha(\Phi)}$,\\
$a=\bigg(1-\frac{8}{9}g\big(\frac{1}{N}(1-\frac{\alpha(\Phi)+1}{\Phi+1})\bigr)\biggr)^{-1}$,\\ $\alpha(\Phi)=\left(N-1+\sqrt{N-1}\right)\sqrt{\Phi}+2(N-1)$, $g(x)=x-\frac{1}{2}x^2$ for complex LMS and $g(x)=x-x^2$ for LMS over the reals and $\eta_{\infty}(\ell)$ is the maximal MSE of the LMS ${\bf W}$ at steady state.
\end{theorem}
Note that $\mu=\frac{1}{3\tr{{\bf R}_\ell}}$ is the maximal update constant suggested in \cite{Feuer_Convergence_1985}. This value is a popular choice as it guarantees the LMS convergence.

Theorem \ref{theorem:main_abstract} has several important implications. Note that $\eta_{\infty}(\ell)$ is in fact the excess MSE which can be tuned by decreasing the step size $\mu$ at the price of a longer convergence time. Moreover, $\eta_{\infty}(\ell)$ approaches zero when the variance of the noise approaches zero. This fact is critical in DSL systems, where the crosstalk interference is much stronger than the background noise. Thus, $\eta_{\infty}(\ell)$ is typically negligible. The Theorem \ref{theorem:main_abstract} implications can be summarized (for high SINR) by:
\begin{equation*}
  \Phi_{\ell+1}\gtrsim c\cdot a^{n_{\ell+1}-n_{\ell}}\Phi_{\ell}.
\end{equation*}
Furthermore, if $\Phi_\ell>1.5N^2+3N$, both  $c$ and $a>1$ are monotonically increasing functions of $\Phi_\ell$.
Hence, updating the preprocessing matrix ${\bf W}_{\rm{P}}$ at (sufficiently large) time $n_{\ell+1}$; i.e., when $c\cdot a^{n_{\ell+1}-n_{\ell}}$ is large enough, will improve the input SINR of the LMS component at times $n>n_{\ell+1}$. Once the preprocessing matrix is updated, the (new) minimal SINR at the input of the LMS is improved accordingly, leading to an increase in $c$ and $a$, and hence to a faster convergence rate of the LMS. Hence, each update of the preprocessing matrix speeds up the convergence of the entire system.
\section{Proof of Theorem \ref{theorem:main_abstract}}
\label{sec:proof}
Due to the length of the proof, we present the proof through a sequence of lemmas, but the proofs of most the lemmas are given in the appendix. Also, for the sake of readability, we begin by a sketch of the proof: we focus on the set of times between two updates of the preprocessing matrix; i.e., the set of times $\mathcal{F}_\ell\cup\{n_{\ell+1}\}$. First, we follow the steps in \cite{Horowitz_Performance_1981,Feuer_Convergence_1985} and characterize the propagation of the MSE at the output of the crosstalk canceler using a propagation matrix ${\bf F}_{\ell}$. Then, Lemma \ref{lemma:the_need_for_l1_norm_of_F} upper bounds the maximal MSE at the output of the Deep-LMS for $n\in\mathcal{F}_{\ell}$ in terms of the condition number of ${\bf R}_\ell$, the norm of the propagation matrix ${\bf F}_{\ell}$ and the maximal MSE at the output of the Deep-LMS at time $n_\ell$. Lemma \ref{lemma:bound_condition_norm1F_MSE_minimal_eig_to_traceR} and Lemma \ref{lemma:upper_bound_on_l1_norm_of_F} use a novel bounding analysis to upper bound each of these terms using the knowledge of the minimal input SINR. Then, we argue that the operation of the Deep-LMS at the update time instance $n_{\ell+1}$ is composed of an additional step of the LMS followed by the normalization $\tilde{\bf D}$. Since a normalization by any diagonal matrix has no effect on the SINR, we conclude the proof by using Lemma \ref{lemma:SINR_MSE_relation} that characterizes the relationship between the input/output SINR and the maximal MSE at the output of the LMS block.

Note that if $\left(c\cdot a^{n_{\ell+1}-n_\ell}\Phi_{\ell}\right)^{-1}+\eta_{\infty}(\ell)\geq 1$ the bound in (\ref{eq:main}) is negative (or zero) and hence trivial. Furthermore, this case is not relevant here because it addresses the case of very low SNR in which the FEXT canceler has no advantages. Thus, in the following we prove Theorem 1 for $\left(c\cdot a^{n_{\ell+1}-n_\ell}\Phi_{\ell}\right)^{-1}+\eta_{\infty}(\ell)< 1$.

For the first lemma, we follow Horowitz et al. \cite{Horowitz_Performance_1981} and Feuer et al. \cite{Feuer_Convergence_1985},\footnote{It should be mentioned that the convergence analysis of the LMS in \cite{Horowitz_Performance_1981,Feuer_Convergence_1985} was based on the assumption that various data vectors were mutually independent. In general uses of LMS this assumption is not true (e.g. in transversal filters). However, in the FEXT cancellation case, this independence assumption is accurate.} and write the propagation of the MSE of the crosstalk canceler coefficients in a recursive form. To that end, certain notations are needed: let ${\bf U}_\ell^H{\bs\Lambda}_\ell{\bf U}_\ell$ be the eigenvalue decomposition of ${\bf R}_{\ell}$, the correlation matrix of the signal after preprocessing (see \ref{eq:R_ell}). Let
\begin{equation}\label{eq:V}
{\bf V}[n]={\bf U}_\ell\left({\bf W}[n]-{\bf W}_{\ell*}\right),
\end{equation}
where
\begin{equation}\label{eq:W*}
{\bf W}_{\ell*}={\bf R}_\ell^{-1}{\bf R}_{ud}
\end{equation}
is the optimal estimator of ${\bf d}$ from ${\bf u}$ and ${\bf R}_{ud}=E\{{\bf u}{\bf d}^H\}$. Denote by ${\bf v}_i[n]$ and $v_{i,j}[n]$ the $i$-th column of matrix ${\bf V}[n]$, and the $j$-element in this vector, respectively. We also define matrix ${\bf S}[n]$, in which the $i,j$ element is
\begin{equation}\label{eq:s_ij}
s_{i,j}[n]=E\{|v_{i,j}[n]|^2\}.
\end{equation}
Let ${\bs\epsilon}_{*}$ be the minimum attainable MSE vector; i.e., the vector that minimizes the cost function $J[n]=\|{\bs \epsilon}[n]\|_{\ell_2}^2$ by choosing ${\bf W}[n]={\bf W}_{\ell*}$. Note that
\begin{equation}\label{eq:epsilon_star_i}
\epsilon_{*i}:=E\{{\bf d}_i\}-{\bf r}_{{\bf u}d_i}^H{\bf R}_\ell^{-1}{\bf r}_{{\bf u}d_i}
\end{equation}
is not affected by the preprocessing matrix and hence does not depend on $n$, where ${\bf r}_{{\bf u}d_i}=E\{{\bf u}d_i^*\}$ and $d_i^*$ is the conjugate of $d_i$.
Denote the eigenvalues of ${\bf R}_\ell$ by
\begin{equation*}
0<\lambda_{\ell,1}\leq\lambda_{\ell,2}\leq\cdots\leq\lambda_{\ell_N}.
\end{equation*}
\begin{lemma}
\label{lemma:MSE_propagation}
The propagation of the MSE of the coefficients of ${\bf W}$ for $n\in\mathcal{F}_{\ell}$ can be written as
\begin{equation}
\label{eq:recursion_on_S}
{\bf S}[n+1]={\bf F}_\ell{\bf S}[n]+4\mu^2{\bs\lambda}_\ell{\bs\epsilon}^T_{*}
\end{equation}
where
\begin{equation}
\label{eq:F}
{\bf F}_\ell=\diag{({\bs \rho}_\ell)}+4\mu^2{\bs\lambda}_\ell{\bs\lambda}_\ell^T,
\end{equation}
${\bs \rho}_\ell:=\left[\rho_{\ell ,1},\rho_{\ell ,2},\cdots,\rho_{\ell,N}\right]$, ${\bs\lambda}_\ell:=\left[\lambda_{\ell ,1},\lambda_{\ell , 2},\cdots,\lambda_{\ell,N}\right]$, $\rho_{\ell ,i}=1-4\mu_\ell\lambda_{\ell , i}+8q\mu_\ell^2\lambda_{\ell ,i}^2$ such that $q=\frac{1}{2}$ for complex LMS and $q=1$ for LMS over the reals. Furthermore, the MSE vector at the output of the Deep-LMS for $n\in\mathcal{F}_{\ell}$ is given by
\begin{equation}
\label{eq:MSE_vector}
{\bs\epsilon}^T[n]={\bs 1}_N^T\tilde{\bf S}[n]+{\bs\epsilon}_{*}^T,
\end{equation}
where $\tilde{\bf S}[n]:={\bs \Lambda}_\ell{\bf S}[n]$.
\end{lemma}
\begin{proof}
The proof for the real case is given in \cite{Feuer_Convergence_1985}. An explicit
computation for the complex case is given in the Appendix.
\end{proof}

As can be seen from (\ref{eq:recursion_on_S}), the eigenvalues of ${\bf F}_\ell$ determine the rate of convergence of ${\bf S}[n]$ (only between the updates of the preprocessing matrix). In \cite{Feuer_Convergence_1985}, it was pointed out that for a real LMS, if $\mu\leq\frac{1}{3\tr{{\bf R}_\ell}}$ all eigenvalues of ${\bf F}$ are less than $1$ and (\ref{eq:recursion_on_S}) converges. However, the analysis in \cite{Feuer_Convergence_1985} is not sufficient to determine the rate of convergence.

From the recursion (\ref{eq:recursion_on_S}) it is obvious that
\begin{equation}
\label{eq:recursion_on_LambdaS}
{\bs\Lambda}_\ell{\bf S}[n+1]={\bs\Lambda}_\ell{\bf F}_\ell{\bs\Lambda}_\ell^{-1}{\bs\Lambda}_\ell{\bf S}[n]+4\mu_\ell^2{\bs\lambda}_\ell^2{\bs\epsilon}_{*}^T,
\end{equation}
where for any vector ${\bf v}$, we use the notation $({\bf v}^2)_j=v_j^2$. Denoting $\tilde{\bf F}_\ell\triangleq{\bs\Lambda}_\ell{\bf F}_\ell{\bs\Lambda}_\ell^{-1}$, (\ref{eq:recursion_on_LambdaS}) can be rewritten as
\begin{equation}
\label{eq:recursion_on_tildeS}
\tilde{\bf S}[n+1]=\tilde{\bf F}_\ell\tilde{\bf S}[n]+4\mu_\ell^2{\bs\lambda}_\ell^2{\bs\epsilon}_{*}^T.
\end{equation}

Note that $\tilde{\bf F}_\ell$ and ${\bf F}_\ell$ are similar matrices and therefore they share their eigenvalues. Combining (\ref{eq:MSE_vector}) and (\ref{eq:recursion_on_tildeS}) yields that for any $n\in\mathcal{F}_\ell$:
\begin{align}
\label{eq:MSE_vector_n_0}
{\bs\epsilon}^T[n]={\bs 1}_N^T\tilde{\bf F}_\ell^{n-n_\ell}\tilde{\bf S}[n_\ell]+4\mu_\ell^2{\bs 1}_N^T\dsum_{k=0}^{n-n_\ell-1}{\tilde{\bf F}_\ell^k}{\bs\lambda}_\ell^2{\bs\epsilon}_{*}^T+{\bs\epsilon}_{*}^T.
\end{align}

In Lemma \ref{lemma:the_need_for_l1_norm_of_F}, we upper bound $\|{\bs\epsilon}[n]\|_\infty$ for $n\in\mathcal{F}_\ell$, where $\|\cdot\|_\infty$ indicates the $\ell_\infty$ norm.
\begin{lemma}
\label{lemma:the_need_for_l1_norm_of_F}
\label{lemma:excess_MSE_bound}
\label{lemma:MSE_propagation_bound}
For $n\in\mathcal{F}_\ell$, if all eigenvalues of ${\bf F}_\ell$ are smaller than $1$, the maximal MSE at the output of the Deep-LMS is given by
\begin{equation}
\label{eq:Maximal_Deep_LMS}
\|{\bs\epsilon}[n]\|_{\infty}\leq \frac{\lambda_{\max}\left({\bf R}_{\ell}\right)}{\lambda_{\min}({\bf R}_{\ell})}\left(\|{\bf F}_\ell\|_1\right)^{n-n_\ell}\|{\bs\epsilon}[n_\ell]\|_{\infty}+\eta_{\infty}(\ell),
\end{equation}
where $\|\cdot\|_1$ indicates the induced $\ell_1$ norm. Furthermore,
\begin{equation*}
\eta_{\infty}(\ell)\leq \left(1+4\mu_\ell^2{\bs\lambda}_\ell^T\left({\bf I}-{\bf F}_\ell\right)^{-1}{\bs\lambda}_\ell\right)\|{\bs\epsilon}_*\|_\infty.
\end{equation*}
\end{lemma}
\begin{proof}
The proof will be given in Appendix \ref{appendix:proof_of_lemmas}.
\end{proof}
Lemma \ref{lemma:bound_condition_norm1F_MSE_minimal_eig_to_traceR} uses Gershgorin circle theorem to develop bounds that involves the eigenvalues of ${\bf R}_\ell$
\begin{lemma}
\label{lemma:bound_condition_norm1F_MSE_minimal_eig_to_traceR}
\begin{enumerate}[(a)]
  \item
\label{item:condition_number_R_upper_bound_by_SINR}
The ($\ell_2$-norm based) condition number of ${\bf R}_\ell$ is upper bounded by:
\begin{equation}
\frac{\lambda_{\max}\left({\bf R}_{\ell}\right)}{\lambda_{\min}({\bf R}_{\ell})}\leq 1+\delta\left(\Phi_\ell\right),
\end{equation}
where $\delta\left(\Phi\right)$ is defined in Theorem \ref{theorem:main_abstract}.
  \item
\label{item:minimal_eigenvalue_to_traceR_lower_bound_by_SINR}
The ratio between the minimal eigenvalue of ${\bf R}_\ell$ to the trace of ${\bf R}_\ell$ is lower bounded by:
\begin{equation}
  \frac{\lambda_{\min}\left({\bf R}_\ell\right)}{\tr{{\bf R}_\ell}}\geq
\frac{1}{N}\bigg(1-\frac{\alpha\left(\Phi_\ell\right)+1}{\Phi_\ell+1}\biggr),
\end{equation}
where $\alpha\left(\Phi\right)$ is defined in Theorem \ref{theorem:main_abstract}.
\end{enumerate}
\end{lemma}
\begin{proof}
See Appendix \ref{appendix:proof_of_lemmas}.
\end{proof}

Lemma \ref{lemma:MSE_propagation_bound} upper bounds the MSE by an expression that involves the $\ell_1$-norm of ${\bf F}_\ell$. In order to get the desired bound, we need to bound $\|{\bf F}_\ell\|_1$ using the minimal input SINR. To that end, we first show that
\begin{equation}
\label{eq:l1_norm_of_F}
\|{\bf F}_\ell\|_1=1-\frac{8}{9}g\left(\frac{\lambda_{\min}\left({\bf R}_\ell\right)}{\tr{{\bf R}_\ell}}\right)
\end{equation}
for $\mu_\ell=\frac{1}{3\tr{{\bf R}_\ell}}$, where $g(x)$ is defined in Theorem \ref{theorem:main_abstract}. Then, using Lemma \ref{lemma:bound_condition_norm1F_MSE_minimal_eig_to_traceR}(\ref{item:minimal_eigenvalue_to_traceR_lower_bound_by_SINR}) and the monotonicity of $g(x)$ in the relevant region, we give a SINR-based upper bound on $\|{\bf F}_\ell\|_1$. Lemma \ref{lemma:upper_bound_on_l1_norm_of_F} summarizes this approach.
\begin{lemma}
\label{lemma:upper_bound_on_l1_norm_of_F}
For $N\geq 2$,
\begin{equation}
\label{eq:upper_bound_on_l1_norm_of_F}
\|{\bf F}_\ell\|_{1}\leq 1-\frac{8}{9}g\left(\frac{1}{N}(1-\gamma(\Phi))\right),
\end{equation}
where $\gamma(\Phi)$ is defined in Theorem \ref{theorem:main_abstract}.
\end{lemma}
\begin{proof}
See Appendix \ref{appendix:proof_of_lemmas}.
\end{proof}

\begin{observation}
\label{observation:eigenvaues_of_F}
It is known that for $\mu_\ell=\frac{1}{3\tr{{\bf R}_\ell}}$ the eigenvalues of ${\bf F}_\ell$ are smaller than $1$ for LMS over the reals \cite{Feuer_Convergence_1985}.
Note that using the fact that $|\lambda({\bf F}_\ell)|\leq\|{\bf F}_\ell\|_{1}$ for any symmetric matrix and that $0<g(x)<1$ for $0\geq x\geq 1$, equation (\ref{eq:l1_norm_of_F}) implies that for $\mu_\ell=\frac{1}{3\tr{{\bf R}_\ell}}$ the eigenvalues of ${\bf F}_\ell$ are smaller than $1$ for the complex case as well.
\end{observation}
Combining Lemma \ref{lemma:MSE_propagation_bound}, Lemma \ref{lemma:bound_condition_norm1F_MSE_minimal_eig_to_traceR}(\ref{item:condition_number_R_upper_bound_by_SINR}) and Lemma \ref{lemma:upper_bound_on_l1_norm_of_F}, and using the definitions of $a$ and $c$ from Theorem \ref{theorem:main_abstract},  yields the following inequality for $n\in\mathcal{F}_\ell$
\begin{equation}
\label{eq:MSE_bound_using_SINR_and_MSE0}
\|{\bs\epsilon}[n]\|_\infty\leq c^{-1}\cdot a^{-\left(n-n_\ell\right)}\|{\bs\epsilon}[n_\ell]\|_\infty+\eta_{\infty}(\ell).
\end{equation}

In order to complete the proof we need to characterize the relation between the SINR and the MSE at the output of the Deep-LMS. To that end, denote the SINR at the $i$-th output of the Deep-LMS by $\Psi_i[n]$ and consider the following lemma:
\begin{lemma}
\label{lemma:SINR_MSE_relation}
\begin{enumerate}[(a)]
  \item \label{item:n_ell_SINR_MSE}
For $n_\ell\in\mathcal{U}$, $\Phi_i[n_\ell]=\epsilon_i^{-1}[n_\ell]$.
  \item \label{item:high_SINR_MSE}
 If $\epsilon_i[n]<1$, the SINR at the output of the Deep-LMS $\Psi_i[n]\geq\epsilon_i^{-1}[n]-1$.
\end{enumerate}
\end{lemma}
\begin{proof}
See Appendix \ref{appendix:proof_of_lemmas}.
\end{proof}

From (\ref{eq:MSE_bound_using_SINR_and_MSE0}) and Lemma \ref{lemma:SINR_MSE_relation}(\ref{item:n_ell_SINR_MSE}) we have that for $n\in\mathcal{F}_\ell$
\begin{equation}
\label{eq:MSE_bound_using_SINR_and_SINR0}
\|{\bs\epsilon}[n]\|_\infty\leq c^{-1}\cdot a^{-\left(n-n_\ell\right)}\Phi_\ell^{-1}+\eta_{\infty}(\ell).
\end{equation}
The operation of the Deep-LMS at the update time $n_{\ell+1}$
is composed of an additional step of the LMS followed by
the normalization by $\tilde{\bf D}$. As the multiplication by a diagonal matrix does not change the SINR, (\ref{eq:MSE_bound_using_SINR_and_SINR0}) can also be used for the evaluation of the SINR at time $n_{\ell+1}$. Thus, the proof of Theorem \ref{theorem:main_abstract} is completed by applying Lemma \ref{lemma:SINR_MSE_relation}(\ref{item:high_SINR_MSE}) on (\ref{eq:MSE_bound_using_SINR_and_SINR0}) at time $n_{\ell+1}$. Note that the condition $\epsilon_i[n]=1$ is equivalent to $\left(c\cdot a^{n_{\ell+1}-n_\ell}\Phi_{\ell}\right)^{-1}+\eta_{\infty}(\ell)> 1$ which was discussed at the beginning of the proof.
\qed
\section{Numerical Results}
\label{sec:numerical_results}
In this section we demonstrate the performance of the Deep-LMS algorithm and
compare it to the performance of the conventional LMS. We also examine the performance of the accelerated versions where the LMS component is replaced by its averaged version.

In \cite{Zanko_Gigabit_2016}, the performance of the Deep-LMS algorithm was tested using model-based channel matrices. In this paper, we used channel matrices measured by BT over a $100$m $10$ pair $0.5$mm copper cable \cite{XXX}. We simulated a typical up-stream scenario with a transmit PSD mask as in \cite{G.9701} and a colored background noise of $-140dBm/Hz$ below $30MHz$ and $-150dBm/Hz$ above $30MHz$ \cite{Strobel_Coexistence_2015}.

An illustration of the eigenvalue spread of the received signal correlation matrix is shown in Fig \ref{fig:non_diagonal_dominant}. The figure also illustrates the non-diagonal dominant structure of the channel matrix at high frequencies.
\begin{figure}[t]
 \centering
  \includegraphics[width=0.5\textwidth]{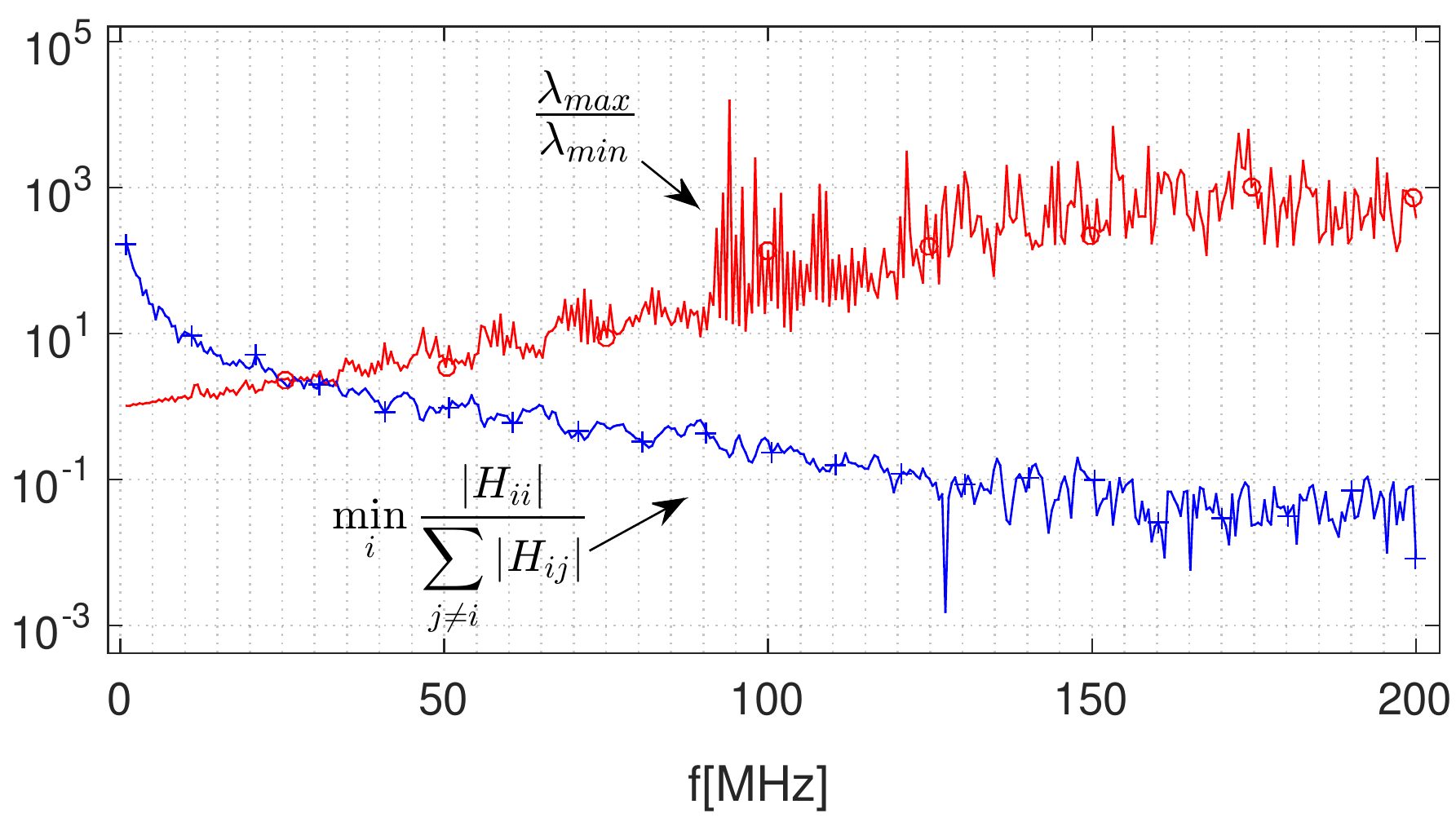}
  \caption{An illustration of the eigenvalue spread of the correlation matrix of the received signal and the non-diagonal dominant structure of the channel matrix at high frequencies.}
  \label{fig:non_diagonal_dominant}
  \label{fig:ill_condition}
\end{figure}
The diagonal dominance is measured by the ratio between the diagonal term and the sum of the absolute value of all non diagonal terms \cite{Bergel_Performance_2013}. Most traditional DSL algorithms require this ratio to be larger than $1$. Thus Fig. \ref{fig:non_diagonal_dominant} demonstrates the need for novel vectoring algorithms that can operate in the G.fast bandwidth.

Before we present the performance of the Deep-LMS algorithm, we need to characterize the choice of update times. Since SINR estimation is an inherent part of the bit-loading component of the DMT physical layer in DSL systems, the set of update times $\mathcal{U}$ can be easily determined in real time using the measured SINR. Thus, we update the preprocessing matrix of the Deep-LMS whenever the SINR is improved by $5$dB or when more than $\tilde{n}$ iterations have elapsed since the last non-SINR based update of the preprocessing matrix.

Fig. \ref{fig:comparison_LMS_and_pre_filter_average_rate} shows the average sum rate per user in each iteration using the Deep-LMS and the conventional LMS. The sum rate of each user was computed by:
\begin{IEEEeqnarray}{rCl}
\label{eq:Gfastrate_formual}
R_t=W\dsum_{k=1}^{K}{\bigg[\log_2\left(1+\rm{SINR}_{i,k}[n]\right)\biggr]_{b_{\max}}}
\end{IEEEeqnarray}
where $\rm{SINR}_{i,k}[n]$ is the SINR at the output of each algorithm at the $i$-th user at the $k$-th frequency, $b_{\max}=12$ is the maximal number of bits per DMT frequency bin as defined in G.fast, \cite{teken}) and $[x]_b=\min\{x,b\}$. As can be seen, the Deep-LMS algorithm converges much faster than the traditional LMS; for example, it reaches $1.65$Gbps in $10^3$ iterations rather than the $3\cdot 10^3$ iterations in the conventional LMS.

Fig. \ref{fig:comparison_LMS_and_pre_filter_average_rate} also depicts the performance of the accelerated LMS that uses filter averaging of \cite{Kushner_Stochastic_1993}. In this algorithm (marked AVG-LMS) the LMS component was updated as in the conventional LMS (\ref{eq:LMS_matrix_recursion}) but the actual filter that was applied to the data at each iteration $n$ was an averaged weight matrix $\sum_{i=0}^{n}{{\theta^{n-i}\bf W}[i]}$ with a forgetting factor of $\theta=0.95$. An analysis of a stochastic approximation with this averaging method can be found for example in \cite{Wang_Weighted_1997}.
This averaging indeed accelerates the convergence, but it is still much slower than the Deep-LMS. Furthermore, the same averaging technique can also be applied to the Deep-LMS (marked AVG Deep-LMS), resulting in even faster convergence.

To further simplify the algorithm, we also tested performance when $\tilde{\bf D}$ was set to ${\bf I}$ in (\ref{eq:the_prefilter}); i.e., when we violated the assumption that the preprocessing matrix was normalized such that the direct effective channel gain of all users was $1$.
As can be seen, the loss due to this simplification is negligible. Hence, we conjecture that this normalization is required mostly for the analysis, but has no significant effect on actual performance.

\begin{figure}[t]
 \centering
  \includegraphics[width=0.5\textwidth]{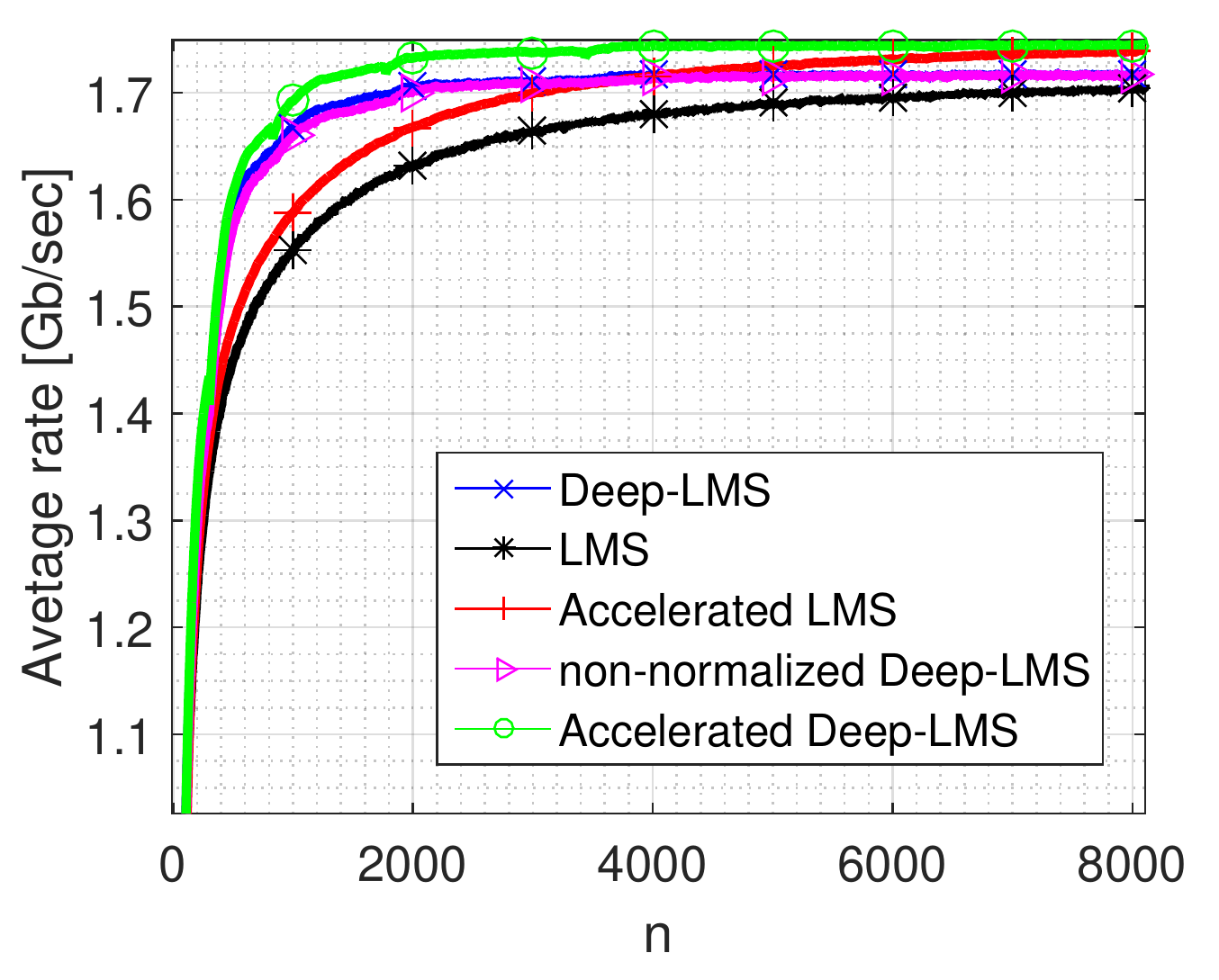}
  \caption{Average (over all users) sum rate comparison (g.fast frequency bins 17MHz-200MHz).}
  \label{fig:comparison_LMS_and_pre_filter_average_rate}
\end{figure}

\begin{figure}[t]
 \centering
  \includegraphics[width=0.5\textwidth]{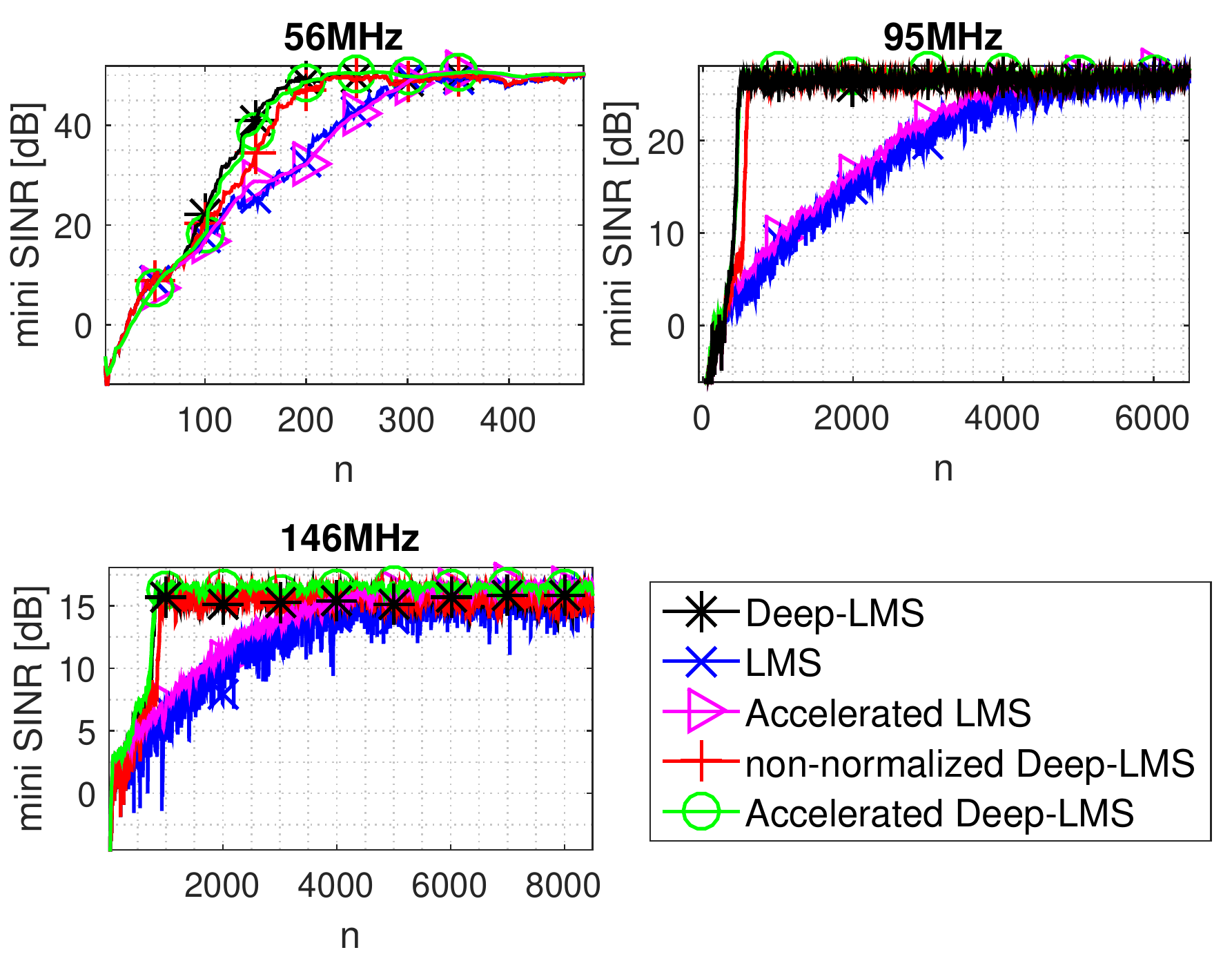}
  \caption{Minimal output $\SINR$ (over all users) comparison at certain frequency bins.}
  \label{fig:comparison_LMS_and_pre_filter_freq_rate}
\end{figure}

To better illustrate the behavior of the different algorithms, Fig. \ref{fig:comparison_LMS_and_pre_filter_freq_rate} shows the minimal SINR at the output of the adaptive crosstalk canceler for each of the algorithms in some specific frequency bins. This figure provides a better understanding of the nature of the proposed Deep-LMS algorithm. It shows that the Deep-LMS starts exactly like the conventional LMS. But, as soon as the traditional LMS manages to improve the SINR above a certain point, the Deep-LMS takes advantage of this improved SINR to produce a significant increase in the convergence rate.

In the second set of simulations, we simulated a near-far scenario. Fig \ref{fig:Gfast_NearFar} illustrates a typical G.fast near-far scenario.
\begin{figure}[t]
 \centering
  \includegraphics[width=0.2\textwidth]{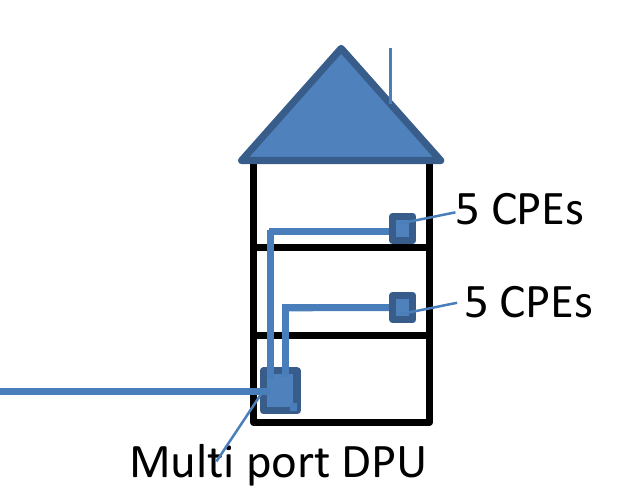}
  \caption{A typical G.fast near far scenario (a FTTB deployment).}
  \label{fig:Gfast_NearFar}
\end{figure}
Taking the same channel matrices as above \cite{XXX} and rewriting them in the form of $5\times 5$ blocks
\begin{equation}\label{eq:H_blocks}
  {\bf H}=\left[\begin{array}{cc}
                  {\bf H}_{11} & {\bf H}_{12} \\
                  {\bf H}_{21} & {\bf H}_{22}
                \end{array}\right],
\end{equation}
we constructed the channel matries for a near-far scenario in which there were $5$ CPEs that were $100$m away from the DPU and $5$ CPEs $200$m away by
\begin{equation*}
  {\bf H}_{NF}=\left[\begin{array}{cc}
                  {\bf I} & {\bs 0} \\
                  {\bs 0} & {\bf H}_{22}
                \end{array}\right]{\bf H}.
\end{equation*}
In this set of simulations, we violated the assumption that the pilots ${\bf d}[n]$ were zero mean Gaussian and used QPSK symbols instead; i.e., each of the elements in ${\bf d}[n]$ was drawn uniformly at random from the set $\{\frac{1}{\sqrt2}(\pm 1+\pm j)\}$.
Fig. \ref{fig:NearFar_Rate} shows the average (achievable) sum rate per user in each iteration using the Deep-LMS and the conventional LMS, where the rate was computed as in (\ref{eq:Gfastrate_formual}). While the algorithms were computed for all the users together, the averaging in the figure was taken separately over the far-users and over the near-users to better illustrate that our analysis in section \ref{sec:analysis} considered the weakest user.
\begin{figure}[t]
 \centering
  \includegraphics[width=0.5\textwidth]{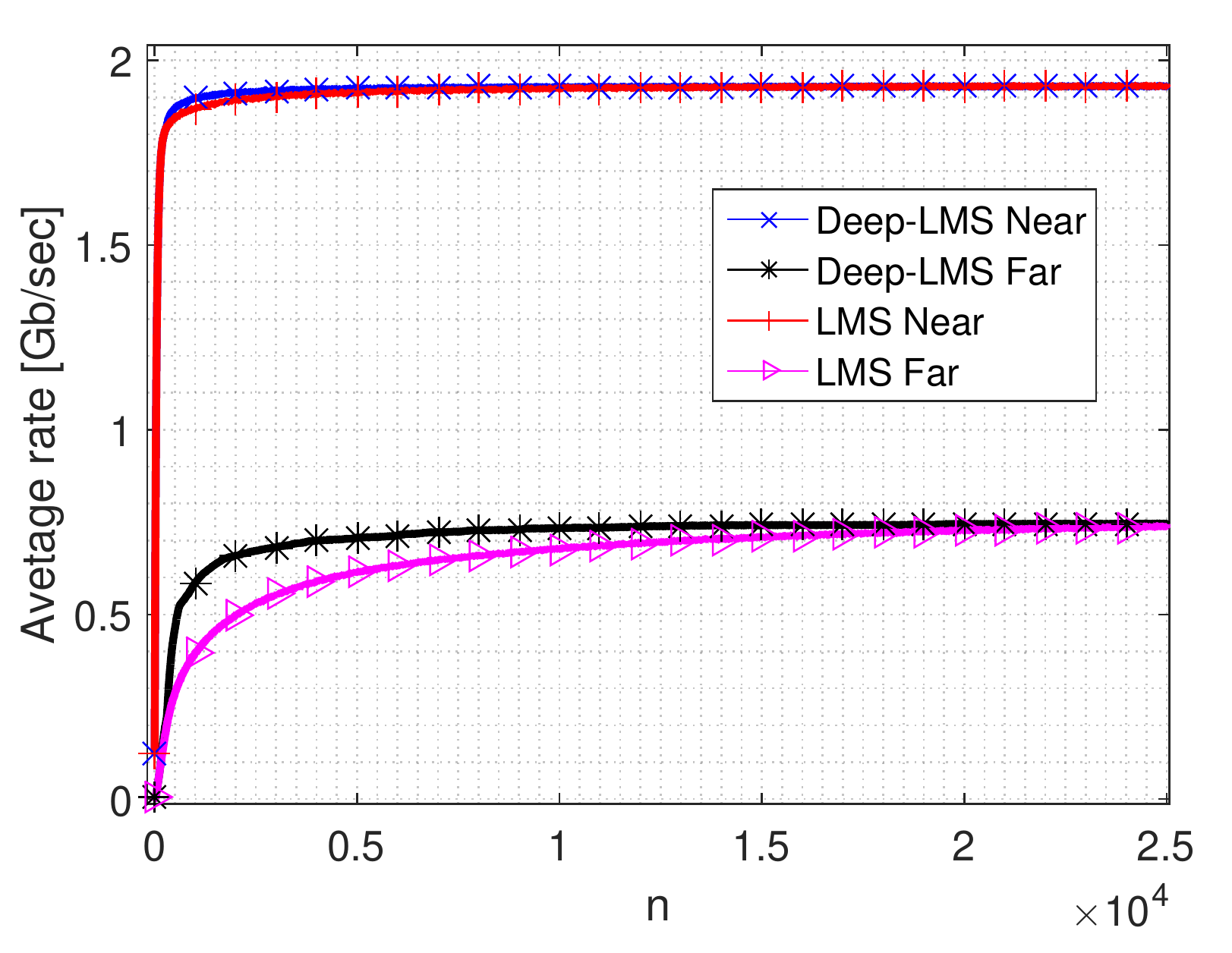}
  \caption{Average sum rate comparison (g.fast frequancy bins 17MHz-200MHz).}
  \label{fig:NearFar_Rate}
\end{figure}
As can be seen, the Deep-LMS algorithm converges much faster than the traditional LMS; for example, the average sum rate of the far users reached $600$Mbps in about one quarter of the time.

\section{Conclusion}
\label{sec:conclusion}
In this paper we presented a new multilayer LMS-based crosstalk canceler for upstream transmission in G.fast systems. The new crosstalk canceler preprocesses the received signal using an adaptive matrix prior to a conventional LMS crosstalk canceler. This preprocessing matrix is initialized by the identity matrix and at any update of the preprocessing matrix is set into the product of the current LMS crosstalk canceler and the current preprocessing matrix. The main goal of the preprocessing matrix is to modify the effective channel matrix into a diagonal dominant structure. We showed that the method can be used to speed up the convergence of the entire system given that the SINR is sufficiently high. Since, the preprocessing matrix is not frequently updated, the complexity of the algorithm is approximately twice the complexity of the conventional LMS. We believe that the proposed method can be used to accelerate LMS algorithms in other contexts as well.
\appendices
\section{Proof of Lemmas}
\label{appendix:proof_of_lemmas}

\begin{proof}[Proof of Lemma \ref{lemma:MSE_propagation}]
Following the steps of \cite{Horowitz_Performance_1981} and adapting to the complex case, we write (\ref{eq:LMS_update_U}) for $n\in\mathcal{F}_\ell$ as
\begin{align}
{\bf W}[n+1]&=\left({\bf I}-2\mu{\bf u}[n]{\bf u}^H[n]\right){\bf W}[n]+2\mu{\bf u}[n]{\bf d}^H[n].
\end{align}
Let $\tilde{\bf u}[n]:={\bf U}{\bf u}[n]$, then using (\ref{eq:V}) we can write that
\begin{align}
&{\bf V}[n+1]={\bf U}\left({\bf W}[n+1]-{\bf W}_*\right)\nonumber\\
&\quad={\bf U}\left(\left({\bf I}-2\mu{\bf u}[n]{\bf u}^H[n]\right){\bf W}[n]+2\mu {\bf u}[n]{\bf d}^H[n]-{\bf W}_{*}\right)\nonumber\\
&\quad={\bf V}[n]-2\mu\tilde{\bf u}[n]\tilde{\bf u}^H[n]{\bf U}{\bf W}[n]+2\mu\tilde{\bf u}[n]{\bf d}^H[n]\nonumber\\
&\quad={\bf V}[n]-2\mu\tilde{\bf u}[n]\tilde{\bf u}^H[n]{\bf U}\left({\bf W}[n]-{\bf W}_{*}+{\bf W}_{*}\right)+\nonumber\\
&\qquad2\mu \tilde{\bf u}[n]{\bf d}^H[n]\nonumber\\
&\quad=\left({\bf I}-2\mu\tilde{\bf u}[n]\tilde{\bf u}^H[n]\right){\bf V}[n]+2\mu\tilde{\bf u}[n]{\bf e}_{*}^H[n],
\end{align}
where the last equality is due to the fact that\\${\bf e}_{*}[n]={\bf d}[n]-{\bf W}_*^H{\bf u}$.

Denote ${\bf C}^{(j)}[n]=E\{{\bf v}_j[n]{\bf v}_j^H[n]\}$. Note that $\tilde{\bf u}[n]$ and ${\bf v}_j[n]$ are independent and that from the orthogonality property we have that $E\{\tilde{\bf u}[n]e_*^H[n]\}=0$.
Hence, taking the expectation in (\ref{eq:v_k1v_k1}), the cross terms are eliminated.
Following \cite{Feuer_Convergence_1985} we can write
\begin{align}
\label{eq:v_k1v_k1}
{\bf C}^{(j)}[n+1]=&E\{{\bf v}_j[n+1]{\bf v}_j^H[n+1]\}\nonumber\\
=& E\left\{\left({\bf I}-2\mu\tilde{\bf u}[n]\tilde{\bf u}^H[n]\right)\times\right.\nonumber\\
&\times\left.{\bf v}_j[n]{\bf v}_j^H[n]\left({\bf I}-2\mu\tilde{\bf u}[n]\tilde{\bf u}^H[n]\right)\right\}+\nonumber\\
&4\mu^2E\{|e_{*j}[n]|^2\tilde{\bf u}[n]\tilde{\bf u}^H[n]\}\nonumber\\
=&{\bf C}^{(j)}[n]-2\mu{\bs \Lambda}{\bf C}^{(j)}[n]-\nonumber\\
&2\mu{\bf C}^{(j)}[n]{\bs \Lambda}+4\mu^2\epsilon_{*j}{\bs \Lambda}+\nonumber\\
&4\mu^2\left(\tr{{\bs\Lambda}{\bf C}^{(j)}[n]}{\bs\Lambda}+{\bs\Lambda}{\bf C}^{(j)}[n]{\bs\Lambda}\right),
\end{align}
where the last term in (\ref{eq:v_k1v_k1}) is derived using the identity
$E\{{\bf z}{\bf z}{\bf A}{\bf z}{\bf z}^H\}=\tr{{\bs\Sigma}{\bf A}}{\bs\Sigma}+{\bs\Sigma}{\bf A}{\bs\Sigma}$, where ${\bf A}$ is a (deterministic) matrix, ${\bs\Sigma}=E\{{\bf z}{\bf z}^H\}$ and ${\bf z}$ is complex Gaussian random vector.

Focusing on the $i$th diagonal element of ${\bf C}^{(j)}[n+1]$:
\begin{align}
\label{eq:c_ii_j}
c_{ii}^{(j)}[n+1]&=\left(1-4\mu\lambda_i+4\mu^2\lambda_i^2\right)c_{ii}(j)[n]\nonumber\\
&+4\mu^2\lambda_i\dsum_{p=1}^N{\lambda_p}c_{pp}(j)[n]+4\mu^2\epsilon_{*j}\lambda_{i}.
\end{align}
Recall that from (\ref{eq:s_ij}) it follows that ${\bf s}_j[n]=\diag\left({\bf C}^{(j)}[n]\right)$ and rewriting (\ref{eq:c_ii_j}) in a matrix form yields the propagation of the MSE of the coefficients of the LMS block. The propagation of the MSE of the output of the LMS in the complex case can be derived following the steps of \cite{Feuer_Convergence_1985}:
first, we write the MSE of the $i$th output of the LMS as
\begin{align}
\epsilon_i[n]=E\{|d_i[n]-{\bf w}_{*i}^H{\bf u}[n]-\left(
{\bf w}_i[k]-{\bf w}_{*i}\right)^H{\bf u}[n]|^2\}.
\end{align}
By noting that $E\{{\bf u}[n]\left(d_i[n]-{\bf w}_{*i}^H{\bf u}[n]\right)\}=0$ and that $E\{{\bf u}\}=0$,
the cross terms are eliminated and we can write that $\Delta[n]=\epsilon_i[n]-\epsilon_{*i}$
\begin{align}
\Delta[n]&=E\{\left(
{\bf w}_i[n]-{\bf w}_{*i}\right)^H{\bf u}[n]{\bf u}^H[n]
\left({\bf w}_i[n]-{\bf w}_{*i}\right)\}\nonumber\\
&=E\{\tr{\left(
{\bf w}_i[n]-{\bf w}_{*i}\right)^H{\bf u}[n]{\bf u}^H[n]
\left({\bf w}_i[n]-{\bf w}_{*i}\right)}\}\nonumber\\
&=\tr{E\{{\bf u}[n]{\bf u}^H[n]
\left({\bf w}_i[n]-{\bf w}_{*i}\right)\left(
{\bf w}_i[n]-{\bf w}_{*i}\right)^H\}}\nonumber\\
&=\tr{E\{{\bf u}[n]{\bf u}^H[n]\}
E\{\left({\bf w}_i[n]-{\bf w}_{*i}\right)\left(
{\bf w}_i[n]-{\bf w}_{*i}\right)^H\}}\nonumber\\
&=\tr{{\bf R}
{\bf U}^HC^{(i)}[n]{\bf U}}\nonumber\\
&=\tr{{\bf U}^H{\bs\Lambda}{\bf U}
{\bf U}^HC^{(i)}[n]{\bf U}}\nonumber\\
&=\tr{{\bs\Lambda}C^{(i)}[n]{\bf U}{\bf U}^H}=\tr{{\bs\Lambda}C^{(i)}[n]}
\end{align}
the lemma now follows.
\end{proof}

\begin{proof}[{\bf Proof of Lemma \ref{lemma:the_need_for_l1_norm_of_F}}]
From (\ref{eq:MSE_vector_n_0}), it is obvious that
\begin{equation}
\|{\bs\epsilon}[n]\|_\infty\leq
\|{\bs 1}_N^T\tilde{\bf F}_\ell^{n-n_\ell}\tilde{\bf S}[n_\ell]\|_\infty+
\tilde{\eta}_\ell\cdot\|{\bs\epsilon}_{*}\|_\infty,
\end{equation}
where $\tilde{\eta}_\ell=4\mu_\ell^2{\bs 1}_N^T\dsum_{k=0}^{n-n_\ell-1}{\tilde{\bf F}_\ell^k}{\bs\lambda}_\ell^2+1$.
Since all entries of $\tilde{\bf F}_\ell$ and $\tilde{\bf S}[n]$ are non-negative real numbers we have that
\begin{equation}
\|{\bs 1}_{_N}^T\tilde{\bf F}_\ell^{n-n_\ell}\tilde{\bf S}[n_\ell]\|_{\infty}=\|\tilde{\bf F}_\ell^{n-n_\ell}\tilde{\bf S}[n_\ell]\|_1.
\end{equation}
The sub-multiplicative property of the induced norm yields
\begin{equation}
\|\tilde{\bf F}_\ell^{n-n_\ell}\tilde{\bf S}[n_\ell]\|_1\leq \|\tilde{\bf F}_\ell^{n-n_\ell}\|_1\|\tilde{\bf S}[n_\ell]\|_1.
\end{equation}
It is easy to verify that since $\tilde{\bf F}_\ell={\bs\Lambda}_\ell{\bf F}_\ell{\bs \Lambda}_\ell^{-1}$
\begin{equation}
\label{eq:tilde_F_to_the_power_of}
\tilde{\bf F}_\ell^{n-n_\ell}={\bs\Lambda}_\ell{\bf F}_\ell^{n-n_\ell}{\bs \Lambda}_\ell^{-1}.
\end{equation}
The fact that for all $i$ and $j$ we have that $\lambda_{\ell,i}>0$ and \\$\left({\bf F}_\ell\right)_{ij}\geq 0$ yields
\begin{align}
\|\tilde{\bf F}_\ell^{n-n_\ell}\|_{1}&\leq
\frac{\lambda_{\max}({\bf R}_{\ell})}{\lambda_{\min}({\bf R}_{\ell})}
\|{\bf F}_\ell^{n-n_\ell}\|_{1}\nonumber\\&\leq\frac{\lambda_{\max}({\bf R}_{\ell})}{\lambda_{\min}({\bf R}_{\ell})}\left(\|{\bf F}_\ell\|_{1}\right)^{n-n_\ell}.
\end{align}

Since all entries of $\tilde{\bf S}$ are non-negative real numbers
we have that
\begin{equation}
\|\tilde{\bf S}[n_\ell]\|_1=\|{\bs 1}_{_{N}}^T\tilde{\bf S}[n_\ell]\|_{\infty}=\|{\bs\epsilon}[n_\ell]\|_{\infty}.
\end{equation}

(\ref{eq:Maximal_Deep_LMS}) follows from the fact that if all eigenvalues of ${\bf F}_\ell$ are smaller than $1$, the first term in the RHS of (\ref{eq:MSE_vector_n_0}) is eliminated and therefore it is obvious that
\begin{equation}\label{eq:explicit_eta_infty_ell}
\eta_{\infty}\left(\ell\right)=4\mu_\ell^2{\bs 1}_N^T\dsum_{k=0}^{n-n_\ell-1}{\tilde{\bf F}_\ell^k}{\bs\lambda}_\ell^2\|{\bs\epsilon}_{*}\|_\infty+\|{\bs\epsilon}_{*}\|_\infty.
\end{equation}

Substituting $\tilde{\bf F}_\ell$ we have that
\begin{align*}
{\bs 1}_N^T\dsum_{k=0}^{n-n_\ell-1}{{\bs\Lambda}{\bf F}^k{\bs\Lambda}^{-1}}{\bs\lambda}^2=\dsum_{k=0}^{n-n_\ell-1}{{\bs\lambda}^T{\bf F}^k{\bs\lambda}}.
\end{align*}
Since ${\bf F}$ is positive definite, we have that ${\bs\lambda}^T{\bf F}^k{\bs\lambda}>0$ for all $k>0$ and $\lambda\in\mathds{R}^N$. Hence,
\begin{equation}
\dsum_{k=0}^{n-n_\ell-1}{{\bs\lambda}^T{\bf F}^k{\bs\lambda}}\leq\dsum_{k=0}^{\infty}{{\bs\lambda}^T{\bf F}^k{\bs\lambda}}=
{\bs\lambda}^T\left(\dsum_{k=0}^{\infty}{{\bf F}^k}\right){\bs\lambda},
\end{equation}
The claim now follows.
\end{proof}

\begin{proof}[{\bf Proof of Lemma \ref{lemma:bound_condition_norm1F_MSE_minimal_eig_to_traceR}}]
Note that since ${\bf u}[n]=\tilde{\bf H}^H[n]{\bf d}[n]+\tilde{\bs\nu}[n]$ and that ${\bf d}[n]$ and ${\bs\nu}[n]$ are independent of each other, the
entries of ${\bf R}_\ell=E\{{\bf u}[n]{\bf u}^H[n]\}$ are given by:
\begin{equation}
\label{R_i_j_color}
\left(R_\ell\right)_{i,j}=\begin{cases}
\|\tilde{\bf h}_{i}\|_2^2+\tilde{\sigma}_{i}^2 & i=j\\
\tilde{\bf h}_i^H\tilde{\bf h}_j+E\{\tilde{\nu}_{i}\tilde{\nu}^*_{j}\} & i\neq j,
\end{cases}
\end{equation}
where $\tilde{\bf h}_i$ is the $i$-th column of $\tilde{\bf H}$. Define the sum of the absolute values of the off-diagonal elements of the row $i$ in ${\bf R}_\ell$ by $B_{\ell,i}\triangleq\dsum_{j\neq i}{|\left(R_{\ell}\right)_{i,j}|}$. Recall that $\lambda_{\ell,i}\;\;i=1,2,\cdots,N$ are the eigenvalues of ${\bf R}_\ell$. From Gershgorin circle theorem, we have that all eigenvalues are in the union of the following regions
\begin{equation}
|\lambda-(\|\tilde{\bf h}_i\|_2^2+\tilde{\sigma}_{i}^2)|\leq B_{\ell,i}\qquad i=1,2,\cdots,N.
\end{equation}
Rearranging the expression we have that
\begin{equation}
\label{eq:lambda_i_Gershgorin_bounds}
\|\tilde{\bf h}_i\|_2^2+\tilde{\sigma}_{i}^2-B_{\ell,i}\leq \lambda\leq\|\tilde{\bf h}_i\|_2^2+\tilde{\sigma}_{i}^2+B_{\ell,i}.
\end{equation}
Since by assumption $\tilde{h}_{ii}[n]=1$, $1\leq \|\tilde{\bf h}_i\|_2^2+\tilde{\sigma}_{i}^2$. Hence,
\begin{equation}
\label{eq:h_tilde_norm_lower_bound}
1-B_{\ell,i}\leq \lambda.
\end{equation}
On the other hand, using (\ref{eq:SINR}) and $\Phi_\ell=\dmin_{i}\Phi_i[n]$ implies that
\begin{equation}
\label{eq:SINR_assumption_1T}
\frac{1}{\Phi_\ell}\geq\frac{1}{\Phi_i[n]}=\dsum_{j\neq k}{|\tilde{h}_{j,k}|^2}+\tilde{\sigma}^2_k.
\end{equation}
By adding $|\tilde{h}_{k,k}|^2=1$ to both sides of (\ref{eq:SINR_assumption_1T}) we have that
\begin{equation}
\label{eq:norm_h_bounds}
1+\frac{1}{\Phi_\ell}\geq |\tilde{h}_{k,k}|^2+\dsum_{j\neq k}{|\tilde{h}_{j,k}|^2}+\tilde{\sigma}^2_k=\|\tilde{\bf h}_k\|_2^2+\tilde{\sigma}^2_k.
\end{equation}
Hence,
\begin{equation}
\label{eq:trace_upper_bound}
\tr{{\bf R}_\ell}=\dsum_{k}{\|\tilde{\bf h}_k\|_2^2+\tilde{\sigma}_k}\leq N\left(1+\frac{1}{\Phi}\right).
\end{equation}
Combining (\ref{eq:lambda_i_Gershgorin_bounds}), (\ref{eq:h_tilde_norm_lower_bound}) and (\ref{eq:norm_h_bounds}) yields
\begin{equation}
1-B_{\ell,i}\leq \lambda \leq 1+\frac{1}{\Phi}+B_{\ell,i}.
\end{equation}
Hence,
\begin{equation}
1-B_\ell\leq\lambda_{\min}\left({\bf R}_\ell\right)\leq
\lambda_{\max}\left({\bf R}_\ell\right)\leq 1+\frac{1}{\Phi}+B_\ell,
\end{equation}
where $B_\ell\triangleq\max_{j}{B_{\ell,j}}$. Thus,
\begin{equation}
\label{eq:bound_lambda_max_min_ratio_using_B_ell}
\frac{\lambda_{\max}\left({\bf R}_\ell\right)}{\lambda_{\min}\left({\bf R}_\ell\right)}\leq
\frac{1+\frac{1}{\Phi}+B_\ell}{1-B_\ell}.
\end{equation}
On the other hand, combining (\ref{eq:h_tilde_norm_lower_bound}) and (\ref{eq:trace_upper_bound}) yields
\begin{equation}
\label{eq:bound_lambda_min_to_traceR_ratio}
\frac{\lambda_{\min}\left({\bf R}_\ell\right)}{\tr{{\bf R}_\ell}}\geq\frac{1-B_\ell}{N\left(1+\frac{1}{\Phi}\right)}.
\end{equation}
\begin{lemma}
\label{lemma:bound_B_ell}
$B_\ell$ is upper bounded by:
\begin{equation*}
B_\ell\leq\frac{\alpha\left(\Phi_\ell\right)}{\Phi_\ell}.
\end{equation*}
\end{lemma}
\begin{proof}[Proof of Lemma \ref{lemma:bound_B_ell}]
From (\ref{R_i_j_color})
\begin{equation}
\label{eq:triangle_ineq_on_off_diag_of_R}
B_{\ell,i}\leq\dsum_{j\neq i}{|\tilde{\bf h}_{i}^H\tilde{\bf h}_{j}|}+\dsum_{j\neq i}{|E\{\tilde{\nu}_{i}\tilde{\nu}^*_{j}\}|}.
\end{equation}
Denote $\mathcal{A}=\{i,j\}$ then for $i\neq j$:
\begin{align}
|\tilde{\bf h}_i^H\tilde{\bf h}_j|&=|\tilde{h}_{i}(i)\tilde{h}_{j}^*(i)+\tilde{h}_{i}(j)\tilde{h}_{j}^*(j)+\dsum_{m\in\mathcal{A}^c}{\tilde{h}_{i}(m)\tilde{h}_{j}^*(m)}|\\
&\leq |\tilde{h}_{i}(i)\tilde{h}_{j}^*(i)|+|\tilde{h}_{i}(j)\tilde{h}_{j}^*(j)|+|\dsum_{m\in\mathcal{A}^c}{\tilde{h}_{i}(m)\tilde{h}_{j}^*(m)}|\\
&\leq |\tilde{h}_{i}(i)\tilde{h}_{j}^*(i)|+|\tilde{h}_{i}(j)\tilde{h}_{j}^*(j)|+\|(\tilde{\bf h}_i)_{\mathcal{A}^c}\|_2\|(\tilde{\bf h}_j)_{\mathcal{A}^c}\|_2,
\end{align}
where the last step is due to the Cauchy-Schwarz inequality and for any set $\mathcal{A}\subseteq\{1,2,\cdots,N\}$ and a vector ${\bf v}$, ${\bf v}_{\mathcal{A}}$ denotes the vector constructed from the entries of ${\bf v}$ indexed by $\mathcal{A}$.
Since
\begin{equation*}
\tilde{h}_{i,i}=\tilde{h}_{i}(i)=\tilde{h}_{j,j}=\tilde{h}_{j}^*(j)=1
\end{equation*}
we have that
\begin{equation}
\label{ineq:h_i_h_j_inner_product_magnitude}
|\tilde{\bf h}_i^H\tilde{\bf h}_j|\leq \|(\tilde{\bf h}_i)_{\mathcal{A}^c}\|_2\|(\tilde{\bf h}_j)_{\mathcal{A}^c}\|_2+|\tilde{h}_{j}^*(i)|+|\tilde{h}_{i}(j)|.
\end{equation}

From the definition of $\Phi_\ell$:
\begin{equation*}
\Phi_{k}=\left(\dsum_{j\neq k}{|\tilde{H}_{j,k}|^2}+\tilde{\sigma}^2_k\right)^{-1}\geq \Phi_\ell.
\end{equation*}
Hence, for any $k$ and $\mathcal{B}\subseteq\{k\}^c$ we have that
\begin{align}
\label{ineq:SINR_and_subset_of_off_diagonal_elements_of_H}
\frac{1}{\Phi_\ell}\geq\dsum_{j\neq k}{|\tilde{H}_{j,k}|^2}+\tilde{\sigma}^2_k=\|(\tilde{\bf h}_{k})_{\{k\}^c}\|^2+\tilde{\sigma}^2_k\geq\|(\tilde{\bf h}_k)_{\mathcal{B}}\|^2.
\end{align}
Therefore, (\ref{ineq:h_i_h_j_inner_product_magnitude}) becomes
\begin{equation}
\label{eq:h_i_h_j_inner_product}
|\tilde{\bf h}_i^H\tilde{\bf h}_j|\leq \frac{1}{\Phi_\ell}+|\tilde{h}_{j}(i)|+|\tilde{h}_{i}(j)|.
\end{equation}
Summing the expression in (\ref{eq:h_i_h_j_inner_product}) over $j\neq i$ yields
\begin{equation}
\dsum_{j\neq i}{|\tilde{\bf h}_i^H\tilde{\bf h}_j|}\leq\frac{N-1}{\Phi_\ell}+\dsum_{j\neq i}{|\tilde{h}_j(i)|}+\dsum_{j\neq i}{|\tilde{h}_i(j)|}.
\end{equation}
From the Cauchy-Schwartz inequality it is known that $\|{\bf x}\|_1\leq \sqrt{n}\|{\bf x}\|_2$, where $n$ is the dimensionality of vector ${\bf x}$. Hence,
\begin{equation}
\label{ineq:h_i_h_j_inner_product_magnitude_with_h}
\dsum_{j\neq i}{|\tilde{\bf h}_i^H\tilde{\bf h}_j|}\leq\frac{N-1}{\Phi_\ell}+\dsum_{j\neq i}{|\tilde{h}_j(i)|}+\sqrt{N-1}\|(\tilde{\bf h}_i)_{\{i\}^c}\|_2.
\end{equation}
By the SINR assumption both $\|(\tilde{\bf h}_i)_{\{i\}^c}\|_2^2$ and $|\tilde{h}_j(i)|^2$ are less than or equal to $\frac{1}{\Phi_\ell}$. Hence, from (\ref{ineq:h_i_h_j_inner_product_magnitude_with_h}) we have that
\begin{equation}
\label{ineq:h_i_h_j_inner_product_magnitude_no_h}
\dsum_{j\neq i}{|\tilde{\bf h}_i^H\tilde{\bf h}_j|}\leq\frac{N-1}{\Phi_\ell}+\frac{N-1}{\sqrt{\Phi_\ell}}+\sqrt{\frac{N-1}{\Phi_\ell}}.
\end{equation}
The second term in (\ref{eq:triangle_ineq_on_off_diag_of_R}) can be bounded by $\frac{N-1}{\Phi_\ell}$ using the fact that  $|E\{\tilde{\nu}_{i}\tilde{\nu}^*_{j}\}|\leq\sqrt{\tilde{\sigma}^2_i\tilde{\sigma}^2_j}\leq\max\{\tilde{\sigma}^2_i,\tilde{\sigma}^2_j\}\leq\frac{1}{\Phi_\ell}$ (where the first inequality is due to the Cauchy–Schwarz inequality).
Lemma \ref{lemma:bound_B_ell} now follows.
\end{proof}
Combining (\ref{eq:bound_lambda_max_min_ratio_using_B_ell}), (\ref{eq:bound_lambda_min_to_traceR_ratio}) and Lemma \ref{lemma:bound_B_ell} concludes the proof of Lemma \ref{lemma:bound_condition_norm1F_MSE_minimal_eig_to_traceR}.
\end{proof}
\begin{proof}[\bf{Proof of Lemma \ref{lemma:upper_bound_on_l1_norm_of_F}}]
First, we use Lemma \ref{lemma:l1_norm_of_F} to establish a closed-form expression of $\|{\bf F}_\ell\|_1$.
\begin{lemma}
\label{lemma:l1_norm_of_F}
For $\mu_\ell=\frac{1}{3\tr{{\bf R}}}$ we have that
\begin{equation*}
\|{\bf F}_\ell\|_{1}=1-\frac{8}{9}g\left(\frac{\lambda_{\min}\left({\bf R}_\ell\right)}{\tr{{\bf R}_\ell}}\right).
\end{equation*}
\end{lemma}
\begin{proof}
Since ${\bf F}_\ell$ is symmetric and its entries are all non-negative, we can write that
\begin{equation}
\label{eq:F_ell_def}
\|{\bf F}_\ell\|_{1}=\dmax_{i}{F_{\ell,i}},
\end{equation}
where $F_{\ell,i}\triangleq\dsum_{j}{\left({\bf F}_\ell\right)_{i,j}}$.
From (\ref{eq:F}) we have that
\begin{equation}
F_{\ell,i}=\rho_{\ell,i}+4\mu_\ell^2\lambda_{\ell,i}\dsum_{j}{\lambda_{\ell,j}}
\end{equation}
Since $\mu_\ell=\frac{1}{3\tr{{\bf R}_\ell}}$, $\tr{R}_\ell=\sum_{j}{\lambda_{\ell,j}}$ and\\$\rho_{\ell,i}=1-\frac{4}{3}\frac{\lambda_{\ell,i}}{\tr{{\bf R}_\ell}}+\frac{8q}{9}\left(\frac{\lambda_{\ell,i}}{\tr{{\bf R}_\ell}}\right)^2$ we can write that
\begin{align}
\label{eq:proof_that_eigF_less_than_1}
F_{\ell,i}&=1-\frac{4}{3}\frac{\lambda_{\ell,i}}{\tr{{\bf R}_\ell}}+\frac{8q}{9}\left(\frac{\lambda_{\ell,i}}{\tr{{\bf R}_\ell}}\right)^2+4\mu_\ell^2\lambda_{\ell,i}\tr{{\bf R}_\ell}\nonumber\\
&=1-\frac{4}{3}\frac{\lambda_{\ell,i}}{\tr{{\bf R}_\ell}}+\frac{8q}{9}\left(\frac{\lambda_{\ell,i}}{\tr{{\bf R}_\ell}}\right)^2+\frac{4}{9}\frac{\lambda_{\ell,i}}{\tr{{\bf R}_\ell}}\nonumber\\
&=1-\frac{8}{9}\left(\frac{\lambda_{\ell,i}}{\tr{{\bf R}_\ell}}-q\left(\frac{\lambda_{\ell,i}}{\tr{{\bf R}_\ell}}\right)^2\right)\nonumber\\
&=1-\frac{8}{9}g\left(\frac{\lambda_{\ell,i}}{\tr{{\bf R}_\ell}}\right).
\end{align}
For complex-LMS, $g(x)=x-\frac{1}{2}x^2$ is monotonically increasing in $x\in[0,1]$ and hence the maxima in (\ref{eq:F_ell_def}) is obtained for $i=1$ and the proof is completed.

In what follows, we show that $\frac{\lambda_{\ell,1}}{\tr{{\bf R}_\ell}}$ is a minimizer of $g\left(\frac{\lambda_{\ell,i}}{\tr{{{\bf R}_\ell}}}\right)$ also for the real case and hence the lemma follows. To that end, it is convenient denote $\psi_{\ell,i}\triangleq\frac{\lambda_{\ell,i}}{\tr{{\bf R}_\ell}}$ for $i=1,2,\cdots,N$. Hence, combining (\ref{eq:F_ell_def}) and (\ref{eq:proof_that_eigF_less_than_1}) results in
\begin{equation}
\|{\bf F}_\ell\|=1-\frac{8}{9}\min_{i}{g\left(\psi_{\ell,i}\right)}.
\end{equation}

Since $g(x)$ is concave, for any $\psi_{\ell,1}\leq \psi_{\ell,k} \leq \psi_{\ell,N}$ $k=2,\cdots,N-1$ we have that $g\left(\psi_{\ell,k}\right)\geq\min\{g\left(\psi_{\ell,1}\right),g\left(\psi_{\ell,N}\right)\}$. Hence,
\begin{equation}
\min_{i}{g\left(\psi_{\ell,i}\right)}=\min\{g\left(\psi_{\ell,1}\right)\;\;,\;\;g\left(\psi_{\ell,N}\right)\}.
\end{equation}
Note that $g(x)$ is monotonically increasing when $x\in[0,0.5]$ and monotonically decreasing when $x\in(0.5,1]$. Hence, it is obvious that if $\psi_{\ell,N}\leq 0.5$ we have that
\begin{equation*}
g\left(\psi_{\ell,1}\right)=\min\{g\left(\psi_{\ell,1}\right),g\left(\psi_{\ell,N}\right)\}. \end{equation*}
On the other hand, since $\dsum_{i=1}^{N}{\psi_{\ell,i}}=1$ and $\psi_{\ell,j}>0$ for all $j$ we have that
\begin{equation*}
\psi_{\ell,N}=1-\dsum_{i=1}^{N-1}{\psi_{\ell,i}}\leq 1-(N-1)\psi_{\ell,1}.
\end{equation*}
Hence, if $\psi_{\ell,N}>0.5$, we have that $g\left(\psi_{\ell,N}\right)\geq g\left(1-\left(N-1\right)\psi_{\ell,1}\right)$. Denote
$\Delta_g=g\left(1-\left(N-1\right)\psi_{\ell,1}\right)-g\left(\psi_{\ell,1}\right)$.
It is easy to see that
\begin{equation*}
\Delta_g=N(N-2)\psi_{\ell,1}\left(\frac{1}{N}-\psi_{\ell,1}\right).
\end{equation*}
Since $N\geq 2$ and $\psi_{\ell,1}\leq\frac{1}{N}\dsum_{j}{\psi_{\ell,j}}=\frac{1}{N}$, we have that $\Delta_g\geq 0$.
Hence, $g\left(\psi_{\ell,1}\right)=\min\{g\left(\psi_{\ell,1}\right),g\left(\psi_{\ell,N}\right)\}$ and the proof is completed.
\end{proof}
Since $N\geq 2$, we have that $\frac{\lambda_{\ell,i}}{\tr{{\bf R}_\ell}}\leq 0.5$. Hence, from the monotonic property of $g(x)$ in the region $x\leq 0.5$ and Lemma \ref{lemma:bound_condition_norm1F_MSE_minimal_eig_to_traceR}(\ref{item:minimal_eigenvalue_to_traceR_lower_bound_by_SINR}) we have that
\begin{equation}
\label{eq:g_psi1_geq_g_Bj}
g\left(\frac{\lambda_{\ell,i}}{\tr{{\bf R}_\ell}}\right)
\geq g\left(\frac{1}{N}\frac{1+\alpha(\Phi)}{\Phi+1}\right).
\end{equation}
Lemma \ref{lemma:upper_bound_on_l1_norm_of_F} now follows.
\end{proof}
\begin{proof}[{\bf Proof of Lemma \ref{lemma:SINR_MSE_relation}}]
From (\ref{eq:r}), (\ref{eq:Deep_LMS_u}) and (\ref{eq:Deep_LMS_x}) we have that
\begin{equation}
\label{eq:Deep_LMS_x_explicit}
  {\bf x}[n]={\bf W}^H[n]{\bf W}_{\rm{P}}^H[n]{\bf H}^H[n]{\bf d}[n]+
             {\bf W}^H[n]{\bf W}_{\rm{P}}^H[n]{\bs \nu}[n].
\end{equation}
Denote the effective channel by $\check{\bf H}[n]\triangleq{\bf H}[n]{\bf W}_{\rm{P}}[n]{\bf W}[n]$ and the effective noise by $\check{\bs\nu}[n]\triangleq{\bf W}^H[n]{\bf W}_{\rm{P}}^H[n]{\bs \nu}[n]$ and denote the variance of $\check{\nu}_i[n]$ by $\check{\sigma}^2_i$. Thus, we can rewrite (\ref{eq:Deep_LMS_x_explicit}) by
\begin{equation}
{\bf x}[n]=\check{\bf H}^H[n]{\bf d}[n]+\check{\bs\nu}[n].
\end{equation}

It is easy to verify that the MSE at the $i$-th output of the Deep-LMS is given by
\begin{equation}
\label{eq:MSE_check}
\epsilon_i[n]=E\{|d_i[n]-x_i[n]|^2\}=|1-\check{h}_{ii}|^2+\dsum_{j\neq i}{|\check{h}_{jj}|^2}+\check{\sigma}^2_i.
\end{equation}
On the other hand, the SINR at the $i$-th output of the Deep-LMS is given by
\begin{equation}
\label{eq:SINR_check}
  \Psi_i[n]=\frac{|\check{h}_{ii}|^2}{\dsum_{j\neq i}{|\check{h}_{jj}|^2}+\check{\sigma}^2_i}.
\end{equation}
Thus, combining (\ref{eq:MSE_check}) and (\ref{eq:SINR_check}) yields
\begin{equation}
\label{eq:SINR_MSE_relation_H_check}
  \Psi_i[n]=\frac{|\check{h}_{ii}|^2}{\epsilon_i[n]-|1-\check{h}_{ii}|^2}.
\end{equation}

For $n_\ell\in\mathcal{U}$ we have that $\check{h}_{ii}[n]=1$ since at these times
${\bf W}[n]$ is initiated to the identity matrix and ${\bf W}_{\rm{P}}$ is normalized such that it has $1$ on the diagonal. Hence, (\ref{item:n_ell_SINR_MSE}) follows.

Obviously, we can always write that
\begin{equation*}
\Psi_i[n]\geq\frac{|\check{h}_{ii}|^2}{\epsilon_i[n]}.
\end{equation*}
Hence, in the case that $|\check{h}_{ii}|\geq1$ (\ref{item:high_SINR_MSE}) follows.
Therefore, in the rest of the proof we assume that $|\check{h}_{ii}|<1$.

We observe that
\begin{equation}
|1-\check{h}_{ii}|^2\geq\left(1-|\check{h}_{ii}|\right)^2.
\end{equation}
Hence, from (\ref{eq:SINR_MSE_relation_H_check}) it is implied that
\begin{align}
\Psi_i[n]&\geq\frac{|\check{h}_{ii}|^2}{\epsilon_i[n]-\left(1-|\check{h}_{ii}|\right)^2}\nonumber\\
&\geq\dmin_{1>\omega>0}\Omega(\omega),
\end{align}
where $\Omega(\omega)=\frac{\omega^2}{\epsilon_i[n]-(1-\omega)^2}$.
It is easy to verify that if \\$\epsilon_i[n]<1$, $\Omega(\omega)$ is minimized when $\omega=1-\epsilon$, hence
\begin{equation}
  \Psi_i[n]\geq \Omega\bigg(1-\epsilon_i[n]\biggr)=\epsilon_i^{-1}[n]-1.
\end{equation}
\end{proof}

\bibliographystyle{IEEEtran}
\bibliography{GfastLMS}
\end{document}

%% file: usepackages.tex
\usepackage{times,amsmath,epsfig}
\usepackage{xspace,latexsym,syntonly}
\usepackage{amssymb}
\usepackage{amsfonts}
\usepackage{textcomp}
\usepackage{graphicx}
\usepackage{epsfig}
\usepackage{dsfont}
\usepackage{enumerate}
\usepackage[sort&compress,square,numbers]{natbib}
\usepackage{paralist}
\usepackage{mathtools} 
\usepackage{amsthm}

%% file: macros.tex
\newcommand{\SINR}{\ensuremath{\hbox{SINR}}}

\newcommand{\diag}{\ensuremath{\hbox{diag}}}

\newtheorem{theorem}{Theorem}
\newtheorem{lemma}[theorem]{Lemma}
\newtheorem{claim}[theorem]{Claim}

\newcommand{\beq}{\begin{equation}}
\newcommand{\eeq}{\end{equation}}
\newcommand{\bea}{\begin{array}}
\newcommand{\ena}{\end{array}}
\newcommand{\bds}{\begin {itemize}}
\newcommand{\eds}{\end {itemize}}
\newcommand{\bdf}{\begin{definition}}
\newcommand{\blm}{\begin{lemma}}
\newcommand{\edf}{\end{definition}}
\newcommand{\elm}{\end{lemma}}
\newcommand{\bthm}{\begin{theorem}}
\newcommand{\ethm}{\end{theorem}}
\newcommand{\bprp}{\begin{prop}}
\newcommand{\eprp}{\end{prop}}
\newcommand{\bcl}{\begin{claim}}
\newcommand{\ecl}{\end{claim}}
\newcommand{\bcr}{\begin{coro}}
\newcommand{\ecr}{\end{coro}}
\newcommand{\bquest}{\begin{question}}
\newcommand{\equest}{\end{question}}

\newcommand{\larrow}{{\larrow}}




%% file: Gfast2015_l1norm.bbl
\begin{thebibliography}{10}
\providecommand{\url}[1]{#1}
\csname url@samestyle\endcsname
\providecommand{\newblock}{\relax}
\providecommand{\bibinfo}[2]{#2}
\providecommand{\BIBentrySTDinterwordspacing}{\spaceskip=0pt\relax}
\providecommand{\BIBentryALTinterwordstretchfactor}{4}
\providecommand{\BIBentryALTinterwordspacing}{\spaceskip=\fontdimen2\font plus
\BIBentryALTinterwordstretchfactor\fontdimen3\font minus
  \fontdimen4\font\relax}
\providecommand{\BIBforeignlanguage}[2]{{%
\expandafter\ifx\csname l@#1\endcsname\relax
\typeout{** WARNING: IEEEtran.bst: No hyphenation pattern has been}%
\typeout{** loaded for the language `#1'. Using the pattern for}%
\typeout{** the default language instead.}%
\else
\language=\csname l@#1\endcsname
\fi
#2}}
\providecommand{\BIBdecl}{\relax}
\BIBdecl

\bibitem{Zanko_Gigabit_2016}
A.~Zanko, I.~Bergel, and A.~Leshem, ``{Gigabit DSL: a Deep-LMS Approach},''
  submitted to EUSIPCO 2016, 2016.

\bibitem{Bingham_Adsl_2000}
J.~Bingham, \emph{{ADSL, VDSL, and multicarrier modulation}}.\hskip 1em plus
  0.5em minus 0.4em\relax Wiley New York, 2000.

\bibitem{Starr_Understanding_1999}
T.~Starr, J.~Cioffi, and P.~Silverman, \emph{Understanding digital subscriber
  line technology}.\hskip 1em plus 0.5em minus 0.4em\relax Prentice Hall PTR,
  1999.

\bibitem{teken}
{ITU-T G.9700}, ``{Fast access to subscriber terminals (G.fast) – Power
  spectral density specification},'' 2014.

\bibitem{Cioffi_Vectored_2006}
J.~Cioffi, M.~Brady, V.~Pourahmad, S.~Jagannathan, W.~Lee, Y.~Kim, C.~Chen,
  K.~Seong, D.~Yu, and M.~Ouzzif, ``{Vectored DSLs with DSM: the road to
  ubiquitous gigabit DSLs},'' 2006.

\bibitem{Huberman_Dynamic_2012}
S.~Huberman, C.~Leung, and T.~Le-Ngoc, ``{Dynamic spectrum management (DSM)
  algorithms for multi-user xDSL},'' \emph{IEEE Communications Surveys \&
  Tutorials}, vol.~14, no.~1, pp. 109--130, 2012.

\bibitem{Bergel_Signal_2013}
I.~Bergel and A.~Leshem, ``{Signal Processing for Vectored Multichannel
  VDSL},'' \emph{Academic Press Library in Signal Processing: Communications
  and Radar Signal Processing}, vol.~2, p. 295, 2013.

\bibitem{Leung_Vectored_2013}
C.~Leung, S.~Huberman, K.~Ho-Van, and T.~Le-Ngoc, ``{Vectored DSL: potential,
  implementation issues and challenges},'' \emph{Communications Surveys \&
  Tutorials, IEEE}, vol.~15, no.~4, pp. 1907--1923, 2013.

\bibitem{Leshem_A_Low_2007}
A.~Leshem and L.~Youming, ``{A low complexity linear precoding technique for
  next generation VDSL downstream transmission over copper},'' \emph{IEEE
  Transactions on Signal Processing}, vol.~55, no.~11, pp. 5527--5534, 2007.

\bibitem{Bergel_Convergence_2010}
I.~Bergel and A.~Leshem, ``{Convergence analysis of downstream VDSL adaptive
  multichannel partial FEXT cancellation},'' \emph{IEEE Transactions on
  Communications}, vol.~58, no.~10, pp. 3021--3027, 2010.

\bibitem{Binyamini_Arbitrary_2012}
I.~Binyamini, I.~Bergel, and A.~Leshem, ``{Arbitrary partial FEXT cancellation
  in adaptive precoding for multichannel downstream VDSL},'' \emph{IEEE
  Transactions on Signal Processing}, vol.~60, no.~11, pp. 5754--5763, 2012.

\bibitem{Cendrillon_Near_2007}
R.~Cendrillon, G.~Ginis, and M.~Moonen, ``{A near-optimal linear crosstalk
  precoder for downstream VDSL},'' \emph{IEEE Transactions on Communications},
  vol.~55, no.~5, pp. 860--863, 2007.

\bibitem{Louveaux_Adaptive_2006}
J.~Louveaux and A.~Van~der Veen, ``{Adaptive DSL crosstalk precancellation
  design using low-rate feedback from end users},'' \emph{IEEE Signal
  Processing Letters,}, vol.~13, no.~11, pp. 665--668, 2006.

\bibitem{Louveaux_Adaptive_2010}
------, ``{Adaptive precoding for downstream crosstalk precancelation in DSL
  systems using sign-error feedback},'' \emph{IEEE Transactions on Signal
  Processing}, vol.~58, no.~6, pp. 3173--3179, 2010.

\bibitem{Cendrillon_Partial_2004}
R.~Cendrillon, M.~Moonen, G.~Ginis, K.~Van~Acker, T.~Bostoen, and P.~Vandaele,
  ``{Partial Crosstalk Cancellation for Upstream VDSL},'' \emph{EURASIP Journal
  on Advances in Signal Processing}, vol. 2004, no.~10, pp. 1--16, 2004.

\bibitem{Homer_Adaptive_2006}
J.~Homer, M.~Gujrathi, R.~Cendrillon, I.~L.~Clarkson, and M.~Moonen,
  ``{Adaptive NLMS partial crosstalk cancellation in digital subscriber
  lines},'' in \emph{Fortieth Asilomar Conference on Signals, Systems and
  Computers, 2006. ACSSC'06.}\hskip 1em plus 0.5em minus 0.4em\relax IEEE,
  2006, pp. 1385--1389.

\bibitem{Cioffi_Gigabit_2003}
J.~Cioffi, M.~Mosheni, A.~Leshem, and Y.~Li, ``{GDSL (Gigabit DSL)},''
  \emph{T1E1.4 contribution T1E1.4/2003-487R1}, 2003.

\bibitem{Widrow_Adaptive_1960}
B.~Widrow and M.~Hoff, ``{Adaptive switching circuits},'' 1960.

\bibitem{Widrow_Adaptive_1971}
B.~Widrow, ``Adaptive filters,'' \emph{Aspects of network and system theory},
  pp. 563--587, 1971.

\bibitem{Horowitz_Performance_1981}
L.~Horowitz and K.~Senne, ``Performance advantage of complex {LMS} for
  controlling narrow-band adaptive arrays,'' \emph{IEEE Transactions on
  Circuits and Systems}, vol.~28, no.~6, pp. 562--576, 1981.

\bibitem{Feuer_Convergence_1985}
A.~Feuer and E.~Weinstein, ``Convergence analysis of {LMS} filters with
  uncorrelated {Gaussian} data,'' \emph{IEEE Transactions on Acoustics, Speech
  and Signal Processing}, vol.~33, no.~1, pp. 222--230, 1985.

\bibitem{Nagumo_A_learning_1967}
J.~Nagumo and A.~Noda, ``A learning method for system identification,''
  \emph{IEEE Transactions on Automatic Control}, vol.~12, no.~3, pp. 282--287,
  1967.

\bibitem{Harris_A_variable_1986}
R.~Harris, D.~Chabries, and F.~Bishop, ``A variable step {(VS)} adaptive filter
  algorithm,'' \emph{IEEE Transactions on Acoustics, Speech and Signal
  Processing}, vol.~34, no.~2, pp. 309--316, Apr 1986.

\bibitem{Kushner_Stochastic_2003}
H.~Kushner and G.~Yin, \emph{Stochastic approximation and recursive algorithms
  and applications}.\hskip 1em plus 0.5em minus 0.4em\relax Springer Science \&
  Business Media, 2003, vol.~35.

\bibitem{Narayan_Transform_1983}
S.~S. Narayan, A.~M. Peterson, and M.~J. Narasimha, ``Transform domain {LMS}
  algorithm,'' \emph{IEEE Transactions on Acoustics, Speech and Signal
  Processing}, vol.~31, no.~3, pp. 609--615, 1983.

\bibitem{Mayyas_Leaky_1997}
K.~Mayyas and T.~Aboulnasr, ``{Leaky LMS algorithm: MSE analysis for Gaussian
  data},'' \emph{IEEE Transactions on Signal Processing}, vol.~45, no.~4, pp.
  927--934, Apr 1997.

\bibitem{Kushner_Stochastic_1993}
H.~J. Kushner and J.~Yang, ``Stochastic approximation with averaging of the
  iterates: Optimal asymptotic rate of convergence for general processes,''
  \emph{SIAM Journal on Control and Optimization}, vol.~31, no.~4, pp.
  1045--1062, 1993.

\bibitem{Ho_A_study_2000}
K.~Ho, ``A study of two adaptive filters in tandem,'' \emph{IEEE Transactions
  on Signal Processing}, vol.~48, no.~6, pp. 1626--1636, 2000.

\bibitem{Ho_Performance_2001}
------, ``Performance of multiple {LMS} adaptive filters in tandem,''
  \emph{IEEE Transactions on Signal Processing}, vol.~49, no.~11, pp.
  2762--2773, 2001.

\bibitem{XXX}
``G.fast: Release of bt cable measurements for use in simulations,'' TD
  2012-11-4A-TC-BT-R2, 10 2012.

\bibitem{G.9701}
{ITU-T G.9701}, ``{Fast Access to Subscriber Terminals (FAST) – Physical
  layer specification},'' 2014.

\bibitem{Strobel_Coexistence_2015}
\emph{{Coexistence of G. fast and VDSL in FTTdp and FTTC Deploy-ments}}.\hskip
  1em plus 0.5em minus 0.4em\relax Proceedings of the 23rd European Signal
  Processing Conference (EUSIPCO-2015), 2015.

\bibitem{Bergel_Performance_2013}
I.~Bergel and A.~Leshem, ``{The performance of zero forcing DSL systems},''
  \emph{IEEE Signal Processing Letters}, vol.~20, no.~5, pp. 527--530, 2013.

\bibitem{Wang_Weighted_1997}
I.~Wang, E.~Chong, and S.~Kulkarni, ``Weighted averaging and stochastic
  approximation,'' \emph{Mathematics of Control, Signals and Systems}, vol.~10,
  no.~1, pp. 41--60, 1997.

\end{thebibliography}
